\documentclass[a4paper,onecolumn,11pt,accepted=2017-05-09]{quantumarticle}
\pdfoutput=1
\usepackage[utf8]{inputenc}
\usepackage[english]{babel}
\usepackage[T1]{fontenc}
\usepackage{amsmath}
\usepackage{hyperref}
\usepackage{graphicx}
\usepackage{float}
\usepackage{subfig}
\usepackage{amsthm}
\usepackage{amssymb}

\usepackage{tikz}
\usepackage{lipsum}
\usepackage[numbers]{natbib}
\newtheorem{theorem}{Theorem}
\newtheorem{lemma}{Lemma}
\newtheorem{proposition}{Proposition}
\newtheorem{definition}{Definition}
\newtheorem{remark}{Remark}
\newtheorem{hypothesis}{Hypothesis}

\begin{document}

\title{Information upper bounds in composite quantum systems}

\author{Zhaoyang Dong}
\affiliation{%
    College of Intelligence and Computing, Tianjin University, No.135, Ya Guan Road, Tianjin 300350, China
}%
\author{Yuexian Hou}%
 \email{yxhou@tju.edu.cn}
 \affiliation{%
    College of Intelligence and Computing, Tianjin University, No.135, Ya Guan Road, Tianjin 300350, China
}%
\author{Chenguang Zhang}
\affiliation{%
    College of Intelligence and Computing, Tianjin University, No.135, Ya Guan Road, Tianjin 300350, China
}%
\affiliation{%
    School of Science, Hainan University, No.58, Renmin Avenue, Haikou 570228, China
}%
\author{Yingjie Gao}
\affiliation{%
    College of Intelligence and Computing, Tianjin University, No.135, Ya Guan Road, Tianjin 300350, China
}%
\author{Dawei Song}
\affiliation{%
    School of Computer Science and Technology, Beijing Institute of Technology, Beijing, 100081, China
}%
\affiliation{
    School of Computing and Communications, The Open University, Milton Keynes, United Kingdom
}
\maketitle

\begin{abstract}
 The intrinsic information of quantum systems refers to the information required to define a quantum state, and may reveal how the nature stores and processes microscopic information. However, there is an evident paradox due to the "\textit{information scale contrast}": Existing analytical results on the information bounds of quantum systems show that the information-carrying capacity of an $n$-qubit system is only $O(n)$. Nonetheless, the intrinsic information content (estimated by, e.g., the number of parameters or the complexity of ontic embedding~\cite{pusey2012reality}) indicates that defining an $n$-qubit system often at least requires information of the order $O(2^n)$. In this paper, we aim to clarify the upper bound of intrinsic information in quantum systems, as well as explain and resolve the aforementioned paradox. Starting with an analysis of the dependence between the Bloch parameters, we take the structural constraints of the quantum state space as a prior and derive the posterior information content of a quantum state through the process of Maximum A Posteriori (MAP) estimation. We analytically prove that the upper bound of the posterior information content of a 2-qubit system is exactly equal to 2. Furthermore, we conjecturally generalized this result to $n$-qubit systems via a convex optimization process and numerical experiments. In summary, our main theoretical observation is that the tight structural constraints among the parameters of a quantum state make them highly interdependent, and thus unable to freely encode information. It turns out that the intrinsic information of an $n$-qubit system defined by posterior information is bounded by $n$ classical bits.
\end{abstract}

\titlepage

\section{\label{introduction}introduction}

Quantum mechanics fundamentally redefined the relationship between information and physical systems. Contrary to classical determinism, the Copenhagen school (Niels Bohr \cite{bohr1928quantenpostulat}, Werner Heisenberg \cite{heisenberg1927anschaulichen}) asserted that quantum states represent epistemic knowledge rather than ontic reality. This informational interpretation finds resonance in modern frameworks like quantum Bayesianism \cite{fuchs2013quantum} and Wheeler's aphorism "it from bit" \cite{wheeler2018information}, albeit with differing metaphysical commitments. In fact, the information interpretation of quantum mechanics must clarify two fundamental questions: Where is the information of the quantum system stored? And what is the quantum information about?

Quantum states are fundamental objects of study in quantum information and quantum computation, and also serve as information carriers in the quantum world~\cite{nielsen2010quantum}. Therefore, understanding the information-carrying and information-spreading capacities of quantum states is of fundamental importance~\cite{holevo1973bounds,PhysRevLett.83.3354,shang2025operator}. However, to date, this fundamental understanding remains unclear and even to some extent confusing~\cite{leifer2014quantum}. The root of this confusion lies in an "information scale contrast" between the intrinsic information content of quantum systems (estimated by, e.g., the complexity of ontic embedding or the number of parameters) and their extrinsic information-carrying capacity (i.e., the scale of information that can be reliably obtained from a quantum system experimentally). For example: In $\psi$-ontic models, exponentially many ontic states are also required to reproduce the predictions of any quantum system, even for a single qubit~\cite{PhysRevA.77.022104}. Moreover, an $n$-qubit system has a parameter scale of $O(2^n)$~\cite{PhysRevA.83.032107}. Nevertheless, under certain typical measurement setups, it can be reasonably argued that the extrinsic information-carrying capacity of an $n$-qubit system is only $O(n)$ (such as the Holevo bound~\cite{holevo1973bounds}). Hardy coined the term "excess baggage" to refer to this phenomenon~\cite{hardy2004quantum}. Given that the parameter scale varies exponentially with the number of qubits, the problem of excess baggage is exacerbated. We attempt to resolve this issue of "excess baggage". Therefore, we should try to clarify the intrinsic information content of quantum systems.

The von Neumann entropy is a potential measure of information. It is a measure of statistical uncertainty in the description of quantum systems, reflecting the complexity of probabilistic mixing in quantum states. However, it fails to capture the internal complexity of pure states. It cannot fully define the intrinsic information content of quantum states. Therefore, it is necessary to conduct in-depth research on other types of intrinsic quantum information. We note that some works by \v{C}aslav Brukner et al. share similar motivations with our work \cite{brukner2003information,brukner2009information,brukner_2025_caehe-cfg98}. \v{C}aslav Brukner and Anton Zeilinger proposed a novel quantum-mechanical information measure that accounts for the fact that the only known characteristic of a quantum system prior to experimental implementation is the probability distribution of observable events. This information measure is widely used in contexts such as Quantum State Estimation~\cite{PhysRevLett.88.130401} and quantum random access codes~\cite{Foundations2002}. But their methods have been controversial~\cite{Shafiee2006,doi:10.1098/rspa.2015.0435}. This information measure is based on a prior axiom in information theory - the maximum entropy principle. When measuring the information of each qubit in an $n$-qubit system, it is assumed that the (prior) reference distribution is uniform. This means the results of $n$ measurements can precisely provide exactly $n$ classical bits of information. However, the measurement processes they defined for quantifying information content do not precisely correspond to those used to determine the state of a quantum system. In other words, we generally cannot determine the state of the measured quantum subsystem by performing the measurement processes they defined. Thus, there exists no clear relationship between the information content they defined and the intrinsic uncertainty level of the quantum system. This does not conform to our aforementioned natural intuition about the definition of information.

In addition, the measurement processes used by \v{C}aslav Brukner to define information did not take into account the intrinsic correlation constraints of quantum systems. For example, in the case of a 2-qubit maximally entangled state, when one qubit is measured, the measurement outcome of the second qubit is already determined. Therefore, if correlations between qubits are considered in \v{C}aslav Brukner's framework for quantifying information, the 2-qubit maximally entangled state should contain only 1 classical bit of information. This example does not satisfy the conclusion given by \v{C}aslav Brukner's information interpretation. Moreover, note that the correlation constraints between quantum measurements may be universally valid. For example, there exists systematic Interdependence between the Bloch parameters of a quantum state. We believe that a solid definition of quantum information content should both align with the natural intuition of information and take into account the prior structural constraints of quantum state space. Based on the above considerations, our definition of information content can be introduced. Note that the information content we define depends simultaneously on non-trivial prior constraints and empirical observations (i.e., measurement results), conforming to the framework of maximum a posteriori (MAP) estimation in a Bayesian perspective. We term it posterior information content.

\section*{Our contribution}

To explore the posterior information content of a quantum system required to determine its state, we define information content based on $\chi^2$ divergence~("CSD" for short), focusing on the decomposition of quantum states and the total information content contained in the components that constitute the quantum state under the Bloch representation. Estimating the total information content necessitates analyzing the structural dependencies inherent in quantum states — A methodological parallel to Quantum State Tomography~(QST)\cite{kalev2015quantum,lloyd2014quantum,Cavalcanti2022postquantum}. When determining the state in a high-dimensional quantum system, multiple local measurements need to be performed. It is usually an iterative procedure, i.e., based on the determined pre-order parameters, the current parameter is determined with MAP estimate. As part of the pre-order parameters in Bloch representation is determined, the value ranges of the post-order parameters are generally no longer $[-1,1]$. That is, the information content obtained from each measurement should be a posterior information that takes into account the constraints between parameters. These constraints arise from the properties of the quantum state such as normalization and positive semi-definiteness. In this paper, we focus on analyzing the structural constraints among Bloch components of quantum states (e.g., Gamel's inequality \cite{gamel2016entangled}). The information content of each Bloch parameter is iteratively estimated based on these constraints until the quantum state is fully determined. 



Through the analysis of the Bloch representation of quantum states, we have analytically proved in this paper that the upper bound of posterior information content for a 2-qubit pure state system is exactly equal to 2. Based on the numerical verification of convex optimization(SDP), the above conclusion can be generalized to the case where n is very large if the computing resources permit, and is valid for the entire quantum space of pure + mixed states. In other words, for an $n$-qubit composite quantum system, despite the large number of parameters (typically $4^n$), the interdependent constraints between these parameters may imply that only $n$ bits of information are ultimately required to fully determine the system. This result establishes an elegant information-theoretic connection between classical and quantum systems, that is an $n$-qubit system contain at most $n$ bits of classical information. 

It should be further emphasized that the Zeilinger–Brukner information interpretation points out that their information measure is essentially different from Shannon information~\cite{PhysRevA.63.022113}, which results in the incompatibility of the information content they derived with other quantum information theories based on Shannon entropy~(in the
form of von Neumann entropy). However, the information measure we defined based on the $\chi^2$ divergence, which defines information content on probability distributions through the concept of posterior information, is essentially a simplification of Shannon information. Therefore, the information measure we proposed is consistent in both classical and quantum regimes. We believe our information measure aligns more closely with the explicitly and systematically articulated shared objective of Carl Friedrich von Weizsäcker’s program~\cite{Görnitz2003}. We believe this work can be combined with other quantum information theories and serve as an information (processing) principle to axiomatically reconstruct quantum mechanics in the future.


\section*{Document structure}
The remainder of this paper is organized as follows. In Sec.~\ref{PRELIMINARIES}, we introduce essential concepts of quantum mechanics reconstructed in accordance with the information principle. In Sec.~\ref{The info post bound of 2}, we analytically derive that the posterior information content upper bound for a 2-qubit system is precisely 2. In Sec.~\ref{the info content of n-qubit}, we extend the analysis to general $n$-qubit systems and propose an upper-bound hypothesis for the posterior information of such systems. The testability of this hypothesis rests on the following theoretical foundation: for any quantum system composed of $n$ qubits, the calculation process can be rigorously formalized as a standard semi-definite programming problem(SDP). Consequently, we can conduct extensive numerical experiments to empirically validate our hypothesis. In Sec.~\ref{conclusion}, we evaluate the hypothesis on the upper bound of posterior information content and present key conclusions regarding composite quantum systems.

\section{\label{PRELIMINARIES}PRELIMINARIES}

\subsection{General representation of 2-qubit pure states}

The Bloch representation offers an intuitive and concise method for describing the states of quantum bit (qubit) in quantum mechanics. The analysis in this work mainly focuses on the information content of the Bloch component of the 2-qubit pure state. Therefore, we first cite some conclusions about quantum pure states. More information on Bloch representations can be found in Appendix \ref{Appendix E} and References \cite{PhysRevLett.113.020402,gamel2016entangled}.

The Bloch representation of any 2-qubit pure state can be obtained by local rotation transformation of the pure state\cite{gamel2016entangled}, that is, the Bloch representation of any 2-qubit pure state
$\psi=(\boldsymbol{\alpha},\boldsymbol{\beta},\boldsymbol{C})$ is
\begin{equation}  
\label{2-qubit_1}
\boldsymbol{\alpha}=O_1\boldsymbol{\alpha'},\ \boldsymbol{\beta}=O_2\boldsymbol{\beta'} ,\ \boldsymbol{C}=O_1\boldsymbol{C'}O_2^T
\end{equation}

where:

\begin{equation} 
\label{2-qubit_2}
\boldsymbol{\alpha'}=\begin{pmatrix}\cos\theta\\0\\0\end{pmatrix},\boldsymbol{\beta'}=\begin{pmatrix}\cos\theta\\0\\0\end{pmatrix},\boldsymbol{C'}=\begin{pmatrix}1&0&0\\0&\sin\theta&0\\0&0&-\sin\theta\end{pmatrix}
\end{equation}
$O_{1}$, $O_{2}\in SO(3)$ are respectively three-dimensional rotation matrices:
\begin{widetext}
    \begin{equation}  
    \label{2-qubit_3}
\begin{aligned}
O_1&=\begin{pmatrix}\cos\phi cosk&\cos\omega sink+\sin\omega sin\phi cosk&\sin\omega sink-\cos\omega sin\phi cosk\\-\cos\phi sink&\cos\omega cosk-\sin\omega sin\phi sink&\sin\omega cosk+\cos\omega sin\phi sink\\\sin\phi&-\sin\omega cos\phi&\cos\omega cos\omega\end{pmatrix}
\\O_2&=\begin{pmatrix}\cos\phi'\cos k'&\cos\omega'\sin k'+\sin\omega'\sin\phi'\cos k'&\sin\omega'\sin k'-\cos\omega'\sin\phi'\cos k'\\-\cos\phi'\sin k'&\cos\omega'\cos k'-\sin\omega'\sin\phi'\sin k'&\sin\omega'\cos k'+\cos\omega'\sin\phi'\sin k'\\\sin\phi'&-\sin\omega'\cos\phi'&\cos\omega'\cos\omega'\end{pmatrix}
\end{aligned}
\end{equation}
\end{widetext}
where $\theta$ can always be in the first quadrant. In other words, $\sin\theta$ and $\cos\theta$ can always be made non-negative. This is because, through local transformations $O_1$ and $O_2$ any angle $\theta'$ outside the first quadrant can be mapped to an angle $\theta$ within the first quadrant.

In the form of Bloch representation, the standard quantum state $\psi$ contains a total of 16 elements and is usually represented as a vector of $16\times 1$. 
\begin{equation}\label{psi}
    {\psi}=\left(b_{1},b_{2},b_{3},\ldots ,b_{i}, \ldots,b_{16}\right)^{T}
\end{equation}
where $b_{i}$ represents the parameter of the Bloch vector. $b_{1} = 1$ is the dummy component and can usually be omitted. 

\subsection{Information divergence}
The information of a random event is quantified by the divergence between the probability distributions describing the system’s state before and after the event occurs. R\'{e}nyi gave a general family of difference measures for probability distribution, which can take different forms depending on the different definitions of its generating function $f$  \cite{renyi1961measures}. 
\begin{equation}
    \label{f_divergence}
	\begin{split}
	D_f(P||Q) \equiv \int_{\Omega } f\left( \frac{dP}{dQ} \right)  dQ
       	\end{split}
\end{equation}
The analysis in this article is primarily based on a specific \( f \)-divergence defined by Eq.~\ref{f_divergence}, the \(\chi^2\)-divergence. This divergence is known to be the upper-bound of Kullback-Leibler (KL) divergence (with natural logarithm base) \cite{popescu2016bounds}, thereby holding statistical relevance. Notably, the KL divergence itself serves as a local approximation of the rigorously defined Fisher-Rao information metric, further reinforcing the utility of the \(\chi^2\)-divergence in statistical analysis. 

The generating function of $\chi^2$ divergence is defined as $f(t) \equiv (t-1)^2$ and it has the following form: 
\begin{equation} 
	\begin{split}
	D_{\chi^2}(P||Q) \equiv \int_{\Omega }  \left( \frac{dP}{dQ} - 1 \right)^2 dQ = \int_{\Omega } \frac{dP^{2}}{dQ} -1
       	\end{split}
\end{equation}
For discrete probability distributions, $P$ and $Q$ correspond to sets of probabilities $\{p_i\}$ and $\{q_i\}$, respectively, satisfying the normalization conditions: $\sum_i p_i = 1 $ and $\sum_i q_i = 1$.
\begin{equation}  
	\begin{split}
	D_{\chi^2}(P||Q) \equiv \sum\nolimits_{{}i} \left( \frac{p_{i}}{q_{i}} -1\right)^{2}  \cdot q_i = \left( \sum\nolimits_{i} \frac{p^{2}_{i}}{q_{i}} \right) -1
       	\end{split}
\end{equation}
Therefore, we propose a rigorous definition of information content based on the $\chi^2$ divergence. 
\begin{definition}[$\chi^2$ divergence information content]
    According to the form of Bloch vector in Eq.\ref{psi}, the $\chi^2$ divergence information content of a $n$-qubit state $\psi$ is defined as:
    \begin{equation}\label{postinfo}
        I_{\chi^2}(\psi)=\sum_i^dD_{X^2}(b_i,h_i)
    \end{equation}
    where $D_{\chi^{2}}(b_{i},h_{i})$ is the $\chi^2$ divergence of the $i$-th component $b_i$ of the Bloch vector $\psi$ relative to the $i$-th component $h_i$ of the vector indicting (prior) reference distributions. According to the relationship between Bloch vector and probability, and the calculation formula of $\chi^2$ divergence \cite{dragomir2002upper}, then
\begin{equation}\label{IC}
    D_{\chi^{2}}(b_{i},h_{i})=\frac{(\frac{b_{i}+1}{2})^{2}}{\frac{h_{i}+1}{2}}+\frac{(\frac{1-b_{i}}{2})^{2}}{\frac{1-h_{i}}{2}}-1
\end{equation}
\end{definition}

It should be emphasized that, considering the distinctive features inherent in composite quantum systems, we start from first principles and rigorously define a posterior information measure by selecting the $\chi^2$ divergence. As is fundamentally established in information theory, "information content" manifests as quantification of surprisal. And likelihood consistency serves as a formalized expression of the intuitive first principle above. The smaller the likelihood of being misestimated, the greater the surprise. In simple terms, a divergence measure $\textbf{d}$ of probability distributions is said to satisfy two-valued likelihood consistency if the following condition holds: For any given distributions $p$ and $q$, if "the likelihood of distribution $p$ being misestimated as $q$ through sampling is greater than that of $q$ being misestimated as $p$", then $ \textbf{d}(p,q) > \textbf{d}(q,p)$. In this sense, it is reasonable to use $\chi^2$ divergence  as a information measure. Beacuse $\chi^2$ divergence measures the possibility of a distribution being "misestimated" as another distribution. Specifically, $\chi^2$ divergence can imply the likelihood consistency~\cite{LikelihoodConsistency}. If we take likelihood consistency as the first principle, then Fisher-Rao divergence (FRD) (or its approximations Kullback-Leibler divergence (KLD)~\cite{kullback1997information,Amari2016} ) can be excluded as information measures. In addition, the $\chi^2$ divergence satisfies "1-bit" requirement, that is: A binary measurement procedure provides at most one classical bit of (posterior) information. Specifically in this context, this means that regardless of our prior knowledge about a Bloch parameter, the measurement process for that Bloch parameter (i.e., the process of determining that Bloch parameter) can provide us with at most one classical bit of (posterior) information. Formally, this corresponds to $\textbf{d}_{\chi^2}(1,0)=1$. We elaborate in detail the more reasons for choosing $\chi^2$ divergence to define information measures. For more detailed analysis, please see the supporting materials.

\section{\label{The info post bound of 2}The information upper bound of pure state in 2-qubit system}

For a general 2-qubit pure state, as described above, it can be shown via Eqs.~\ref{2-qubit_1}--\ref{2-qubit_3} that once certain components of the 16-dimensional vector are specified, the remaining components are uniquely determined. In extreme cases, for example, if $b_2, \ldots, b_{15}$ are known, $b_{16}$ can be directly computed. This demonstrates that the 16-dimensional components of the Bloch vector are not mutually independent.
Components that can be uniquely determined by pre-order components contribute no additional informational. Such redundancies should be eliminated when computing the system's information content. Therefore, we give the posterior information definition based on Eq.\ref{postinfo}:
\begin{definition}[The posterior $\chi^2$ divergence information content]
   For the Bloch vector of the 2-qubit pure state, starting from $b_1$ and based on the determined components, the possible maximum value $b_{i}^{max}$ and minimum value $b_{i}^{min}$ of the current component $b_i$ are calculated. According the maximum entropy estimation principle, $(b_i^{max}+b_i^{min})/2$ is the reference distribution $h_i$ of $b_i$. 
    
    The posterior $\chi^2$ divergence information content of the Bloch vector of 2-qubit system is
    \begin{equation}
    I_{\chi^{2},post}({\psi})=\sum_{i=2}^{16}D_{\chi^{2}}(b_{i},\frac{b_{i}^{max}+b_{i}^{min}}{2})
\end{equation}
\end{definition}
The analysis in this paper does not rely on the prior information content, so please refer to Appendix \ref{Appendix A} for the definition of the prior information content.
\begin{remark}
   When estimating the current parameter, the posterior information adjusts the prior of the current parameter according to the estimated value of the prior parameter, and uses the maximum a posterior estimate (MAP) to determine the current parameter; this process is iterative. When calculating the information, the median point of the value interval of each current parameter is successively determined as the reference distribution based on the estimated value of the prior parameter.
\end{remark}
In addition, the $\chi^2$ divergence determined by the Bloch component satisfies the following Pythagorean relationship.

\begin{lemma}[Pythagorean properties]\label{Pythagorean properties}
    Assume that $t_1$,$t_2$,$t_3$ are components of the Bloch vector, and $t_1^2 = t_2^2 + t_3^2$. then:
\begin{equation}
    D_{X^2}(t_1,0)\equiv D_{X^2}(t_2,0)+D_{X^2}(t_3,0)
\end{equation}
\end{lemma}
\begin{proof}
    The proof is given in Appendix \ref{Appendix B}.
\end{proof}
Lemma \ref{Pythagorean properties} allows us to calculate the information content of each Bloch component in pieces.

A key motivation for quantifying posterior information through pure state space structures lies in their inherently quantum nature: pure states contain no classical uncertainty (or equivalently, no information loss) induced by statistical mixing of probabilities. 

According to Eqs.~\ref{2-qubit_1}--\ref{2-qubit_3}, for any pure state of a 2-qubit system, the only constraint is the equality of the norms of the Bloch vectors \(\boldsymbol{\alpha}\) and \(\boldsymbol{\beta}\), i.e., \(\|\boldsymbol{\alpha}\| = \|\boldsymbol{\beta}\|\). No additional constraints are required between them, that is, their reference distribution is always 0 and the posterior information content is equal to the prior information content. Therefore only the components in matrix $\boldsymbol{C}$ need to be discussed. 

As a result, we have the following lemma.

\begin{lemma}[Unique determination]\label{Unique determination}
For a Bloch vector of pure state in 2-qubit system $\psi=(\boldsymbol{\alpha}, \boldsymbol{\beta}, \boldsymbol{C})$ with $\cos^2\phi\cos^2k\neq0$ or $\cos^2\phi'\cos^2k'\neq0$, when $\boldsymbol{\alpha}$, $\boldsymbol{\beta}$ and the first two components $c_{11}$, $c_{12}$ in the matrix $\boldsymbol{C}$ are determined, then the remaining components of this Bloch vector are uniquely determined.
\end{lemma}
\begin{proof}
    The proof is given in Appendix \ref{Appendix C}.
\end{proof}
The above lemma \ref{Unique determination} is for the general case. In particular, when $\cos^2\phi\cos^2k=0$ or $\cos^2\phi'\cos^2k'=0$, it corresponds to several special cases. 
In these cases, the structure constraints of the 2-qubit system will be weakened, some components in the quantum state will be uniquely determined by $\boldsymbol{\alpha}$ and $\boldsymbol{\beta}$, and only two additional components (\{$c_{12}, c_{13}$\} or \{$c_{21}, c_{22}$\}) need to be determined. When calculating the information content, there are only $(\omega+\omega')$ in the parametric equation, so their reference distribution is all 0. Obviously, it can be concluded that their posterior information content is less than 2. Therefore, we will not discuss these special cases in detail in this article.

Expand and calculate the matrix $\boldsymbol{C}$. Obviously, the sign of the component $\cos \phi$, $\cos{\phi'}$, $\cos {k}$, $\cos {k'}$ does not affect the determination of the components in $\boldsymbol{C}$. In particular, the first two components $c_{11}$ and $c_{12}$ in the $\boldsymbol{C}$ matrix can be represented by components:

\begin{equation}\label{c11}
\resizebox{.95\hsize}{!}{$c_{11}=\cos\phi coskcos\phi^{\prime}cosk^{\prime}-\sin\theta\sqrt{1-\cos^2\phi cos^2k}\sqrt{1-\cos^2\phi^{\prime}cos^2k^{\prime}}\cos(\omega+\omega^{\prime}+\omega_1+\omega_1^{\prime})$}
\end{equation}
\begin{equation}\label{c12}
\resizebox{.95\hsize}{!}{$c_{12}=-\cos\phi coskcos\phi'cosk'+\sin\theta\sqrt{1-\cos^2\phi cos^2k}\sqrt{1-\cos^2\phi'cos^2k'}\sin(\omega+\omega'+\omega_1+\omega_2')$}
\end{equation}
where

\begin{equation}
    \resizebox{.98\hsize}{!}{$\begin{aligned}&\sin\omega_1=\frac{\sin k}{\sqrt{\sin^2k+\sin^2\phi\cos^2k}},\sin\omega_1^{'}=\frac{\sin k^{'}}{\sqrt{\sin^2k+\sin^2\phi^{'}\cos^2k^{'}}},\sin\omega_2^{'}=\frac{\cos k^{'}}{\sqrt{\cos^2k^{'}+\sin^2\phi^{'}\sin^2k^{'}}}\\&\cos\omega_1=\frac{\sin\phi\text{cos}k}{\sqrt{\sin^2k+\sin^2\phi\text{cos}^2k}},\cos\omega_1^{'}=\frac{\sin\phi^{'}\cos k^{'}}{\sqrt{\sin^2k^{'}+\sin^2\phi^{'}\cos^2k^{'}}},\cos\omega_2^{'}=\frac{\sin\phi^{'}\sin k^{'}}{\sqrt{\cos^2k^{'}+\sin^2\phi^{'}\sin^2k^{'}}}\end{aligned}$}
\end{equation}

In order to facilitate subsequent analysis, we let $\mu=\cos\theta$, $a_1=\cos\phi\cos k$, $a_2=\cos\phi\sin k$, $a_3=\sin\phi$, $b_1=\cos\phi'\cos k'$, $b_2=\cos\phi'\sin k'$, $b_3=\sin\phi'$. The form of $\boldsymbol{\alpha}$ and $\boldsymbol{\beta}$ become as follows:
\begin{equation}
    \boldsymbol{\alpha}=\left(\begin{array}{c}{\mu\cos\phi\cos k}\\{-\mu\cos\phi\sin k}\\{\mu\sin\phi}\\\end{array}\right)=\left(\begin{array}{c}{\mu a_{1}}\\{-\mu a_{2}}\\{\mu a_{3}}\\\end{array}\right)
\end{equation}
\begin{equation}
    \boldsymbol{\beta}  = \left( {\begin{array}{*{20}{c}}
{\mu \cos \phi '\cos k'}\\
{-\mu \cos \phi '\sin k'}\\
{\mu \sin \phi '}
\end{array}} \right) = \left( {\begin{array}{*{20}{c}}
{\mu {b_1}}\\
{-\mu {b_2}}\\
{\mu {b_3}}
\end{array}} \right)
\end{equation}
$c_{11}$ and $c_{22}$ in Eq.\ref{c11},\ref{c12} can be represented as:
\begin{equation}
    \begin{aligned}
        c_{11} &= a_{1} b_{1} 
        - \sqrt{(1-\mu^{2})(1-a_{1}^{2})(1-b_{1}^{2})} \times \cos(\omega + \omega' + \omega_{1} + \omega_{1}') \\
        c_{12} &= -a_{1} b_{2} 
        + \sqrt{(1-\mu^{2})(1-a_{1}^{2})(1-b_{2}^{2})} \times \sin(\omega + \omega' + \omega_{1} + \omega_{2}')
    \end{aligned}
\end{equation}
Obviously, if we consider the parameters $\omega$ and $ \omega'$ as variables and the other parameters as constants. The upper and lower bounds of components $c_{11}$ and $c_{22}$ can be represented as:
\begin{equation}\label{c11,c12}
    \begin{aligned}
&c_{11}^{max} \begin{aligned}=a_1b_1+\sqrt{(1-\mu^2)(1-a_1^2)(1-b_1^2)}\end{aligned}  \\
&c_{11}^{min} =a_1b_1-\sqrt{(1-\mu^2)(1-a_1^2)(1-b_1^2)}  \\
&c_{12}^{max} =-a_1b_2+\sqrt{(1-\mu^2)(1-a_1^2)(1-b_2^2)}  \\
&c_{12}^{min} =-a_1b_2-\sqrt{(1-\mu^2)(1-a_1^2)(1-b_2^2)} 
\end{aligned}
\end{equation}
the reference distribution $c_{11}^{avg}$ and $c_{12}^{avg}$ is:
\begin{equation} \begin{aligned}c_{11}^{avg}&=\frac{c_{11}^{max}+c_{11}^{min}}2=a_1b_1\\c_{12}^{av g}&=\frac{c_{12}^{max}+c_{12}^{min}}2=-a_1b_2\end{aligned}
\end{equation}

In order to prove the upper bound of the posterior information content of the 2-qubit system, we will construct the following lemma:

\begin{lemma}\label{mu=0}
    For any 2-qubit pure state Bloch vectors $\psi$ with $\mu=0$.
    \begin{equation}
        I_{\chi^{2},post}(\psi)=2
    \end{equation}     
\end{lemma}
\begin{proof}
    When $\mu=\mathrm{cos}\theta=0$, $|\mathrm{sin}\theta|=1$. The quantum state $\psi$ is in the maximum entangled state. According to Eq.~\ref{2-qubit_2}, the correlation matrix $\boldsymbol{C}$ must be an orthogonal matrix and the modulus of each row and column is 1 (because $\boldsymbol{C}=O_{1}\boldsymbol{C}_{0}O_{2}^{\mathrm{T}}$ is obtained by multiplying three orthogonal matrices). After the first two rows of the orthogonal matrix $\boldsymbol{C}$ are determined, the third row is naturally determined, so no posterior information is provided. And since $\boldsymbol{\alpha}$ and $\boldsymbol{\beta}$ are both $\boldsymbol{0}$ and also do not provide any information, then according to Lemma \ref{Pythagorean properties}.
    \begin{equation}
        I_{\chi^{2},post}(\psi)=1^2+|\mathrm{sin}\theta|^2=2
    \end{equation} 
\end{proof}

\begin{lemma}\label{mu=1}
    For any 2-qubit pure state Bloch vectors $\psi$ with $|\mu|=1$,
    \begin{equation}
        I_{\chi^{2},post}(\psi)=2
    \end{equation}  
\end{lemma}

\begin{proof}
    When $|\mu|=1$, the quantum state is in a tensor product state. $\boldsymbol{\alpha}$ and $\boldsymbol{\beta}$ in the Bloch vector cannot constrain each other, so the information provided by $\boldsymbol{\alpha}$ and $\boldsymbol{\beta}$ is all posterior information. For given $\boldsymbol{\alpha}$ and $\boldsymbol{\beta}$, the parameters $\phi,k,\phi^{\prime},k^{\prime}$ can be determined, and only two parameters $\omega$ and $\omega^{\prime}$ are undetermined at this time. Since $|\mu|=|\mathrm{cos}\theta|=1$, then $\mathrm{sin}\theta=0$. $C'$ can be obtained
    \begin{equation}
        C'=\begin{pmatrix}1&0&0\\0&0&0\\0&0&0\end{pmatrix}
    \end{equation}
By directly calculating each element in the correlation matrix $\boldsymbol{C}$, it is not difficult to find that $\omega$ and $\omega^{\prime}$ do not affect any components in $\boldsymbol{C}$. Therefore, the upper and lower bounds of each element in $\boldsymbol{C}$ are equal and no posterior information is provided. At this time, the information held by Bloch vector is completely provided by $\boldsymbol{\alpha}$ and $\boldsymbol{\beta}$. Then
\begin{equation}
    I_{\chi^{2},post}(\psi)=2
\end{equation}
\end{proof}

\begin{lemma}\label{mu=(0,1)}
    For any 2-qubit pure state Bloch vectors $\psi$ with $|\mu|\in(0,1)$,
    \begin{equation}
         I_{\chi^{2},post}(\psi)<2
    \end{equation}
\end{lemma}
\begin{proof}
    Without loss of generality, assume $c_{11}^{avg}$ and $c_{12}^{avg}$ are both positive. $\theta$ is in the first quadrant, then $\mu=\cos\theta\in(0,1)$. Directly calculate the posterior information content of the quantum state at this time: 
     \begin{equation}\label{finalpost}
    \begin{aligned}
    I_{\chi^{2},\text{post}}(\psi) = & 2D_{\chi^{2}}(\mu,0)+ D_{\chi^{2}}(c_{11}^{\text{max}},c_{11}^{avg}) + D_{\chi^{2}}(c_{12}^{\text{max}},c_{12}^{a\nu g})
    \end{aligned}
    \end{equation}
    For the convenience of proof, let $\delta=(1-\mu^{2})(1-a_{1}^{2})(1-b_{1}^{2})$ in Eq.\ref{c11,c12}. According to Eq.\ref{postinfo}, calculate the posterior information content:
    \begin{widetext}
        \begin{equation}
        \resizebox{.95\hsize}{!}{$\begin{aligned}
        D_{\chi^{2}}(c_{11}^{max},c_{11}^{avg})& =\frac{(\frac{a_{1}b_{1}+\sqrt{\delta}+1}{2})^{2}}{\frac{a_{1}b_{1}+1}{2}}+\frac{(\frac{1-a_{1}b_{1}-\sqrt{\delta}}{2})^{2}}{\frac{1-a_{1}b_{1}}{2}}-1\quad  \\
        &=\frac{(a_1b_1+1)^2+\delta+2(a_1b_1+1)\sqrt{\delta}}{2(a_1b_1+1)}+\frac{(1-a_1b_1)^2+\delta-2(1-a_1b_1)\sqrt{\delta}}{2(1-a_1b_1)}-1 \\
        &=\frac{(a_{1}b_{1}+1)(1-a_{1}^{2}b_{1}^{2})+(1-a_{1}b_{1})(1-a_{1}^{2}b_{1}^{2})+2\delta}{2(1-a_{1}^{2}b_{1}^{2})}-1 \\
        &=\frac{\delta}{1-a_{1}^{2}b_{1}^{2}}
        \end{aligned}$}
        \end{equation}
    \end{widetext}
    Similarly, $D_{\chi^{2}}(c_{12}^{\text{max}},c_{12}^{avg})$ can perform similar calculations. The posterior information content in Eq.\ref{finalpost} can be organized into the following form. 
    \begin{equation}
        \begin{array}{l}
        {I_{{\chi ^2},post}}(\psi ) = 2{D_{{\chi ^2}}}(\mu ,0) + \frac{{(1 - {\mu ^2})(1 - a_1^2)(1 - b_1^2)}}{{1 - a_1^2b_1^2}} + \frac{{(1 - {\mu ^2})(1 - a_1^2)(1 - b_2^2)}}{{1 - a_1^2b_2^2}}
        \end{array}
    \end{equation}
    
    Let $x_{1}=a_{1}^{2}$,$y_{1}=b_{1}^{2}$,$y_{2}=b_{2}^{2}$. Then, this problem is organized into a discussion of the following functions:
    \begin{equation}
        \begin{array}{l}
             f(x_1,y_1,y_2)= 2D_{\chi^2}(\mu,0)+\frac{(1-\mu^2)(1-x_1)(1-y_1)}{1-x_1y_1}+\frac{(1-\mu^2)(1-x_1)(1-y_2)}{1-x_1y_2}
        \end{array}
    \end{equation}
    where $x_{1},y_{1},y_{2}\in[0,1]$, and $0\leq y_{1}+y_{2}\leq1$. The partial derivative of the function $f(x_1,y_1,y_2)$ with respect to $x_1$. 
    \begin{equation}
    \begin{array}{c}
         \frac{\partial f(x_1,y_1,y_2)}{\partial x_1}=\frac{-(1-\mu^2)(1-y_1)^2}{(1-x_1y_1)^2}+\frac{-(1-\mu^2)(1-y_2)^2}{(1-x_1y_2)^2}
    \end{array}
    \end{equation}
    Obviously, $\frac{\partial f(x_1,y_1,y_2)}{\partial x_1}\leq0$, the equation holds if and only if $y_{1}=y_{2}=1$. Due to the constraint $0\leq y_{1}+y_{2}\leq1$, the equation is not satisfied within the current range of values. Therefore, $f(x_1,y_1,y_2)$ is monotonically decreasing in the $x_1$ direction. This means that for the optimization problem of $f(x_1,y_1,y_2)$, only need to consider the optimal value of $f(0,y_1,y_2)$ in the extreme case of $x_1=0$. That is:
    \begin{equation}
        f(0,y_1,y_2)=2D_{\chi^2}(\mu,0)+(1-\mu^2)(2-y_1-y_2)
    \end{equation}
    Consider the two directions $y_1$ and $y_2$:
    \begin{equation}
        \frac{\partial f(0,y_1,y_2)}{\partial y_1}=\frac{\partial f(0,y_1,y_2)}{\partial y_2}=-(1-\mu^2)<0
    \end{equation}
    At this time, $f(0,y_1,y_2)$ is monotonically decreasing in both $y_1$ and $y_2$ directions. Therefore, when $y_1=y_2=0$, $f(0,y_1,y_2)$ has a maximum value, and the maximum value $\boldsymbol{g}(\mu)$ is: 
    \begin{equation}
        \boldsymbol{g}(\mu)=2D_{X^2}(\mu,0)+2(1-\mu^2)=2
    \end{equation}
    where $\mu\in(0,1)$. Based on the above analysis, it can be proved that in the current range, $f(x_1,y_1,y_2)\leq2$ holds. That is
    \begin{equation}
        I_{{\chi ^2},post}(\psi )\leq2
    \end{equation}
\end{proof}
Lemma \ref{mu=0}, lemma \ref{mu=1} and lemma \ref{mu=(0,1)} cover all Bloch vectors of 2-qubit pure state system, the following theorem has been proven to be true.

\begin{theorem}\label{the upper 2-qubit}
    The upper bound on the posterior information content of a 2-qubit pure state system is exactly equal to 2.
\end{theorem}

\section{\label{the info content of n-qubit}The information upper bound of pure state in $n$-qubit system}

We proceed to generalize Theorem~\ref{the upper 2-qubit} to $n$-qubit systems. Owing to the inherent complexity of their analysis, no tractable analytical solution has yet been established. To address this, we formulate a hypothesis and validate it through systematic numerical investigations.Based on this hypothesis, we have found that regardless of the number of qubits (denoted as $n$), the process of calculating information content using our method can always be expressed as a convergent semidefinite programming (SDP) First, we propose the following hypothesis:

\begin{hypothesis}[The posterior information content upper bound]\label{the upper n-qubit}
    The upper bound on the posterior information content of a $n$-qubit pure state system is exactly equal to $n$.
\end{hypothesis}

The process of calculating the 2-qubit information upper bound in Lemma \ref{mu=(0,1)} actually corresponds to a standard Semi-definite Programming (SDP). Then, we propose the following proposition:

\begin{proposition}[Semi-definite Programming]\label{SDP}
    For any positive integer $n$, any (standard) $n$-qubit system, its first $k-1$ parameters $b_i,\;1 \le i \le k - 1 \le {4^n}$ are given in the natural index order of the standard Bloch representation. Assume ${b_k}$ is the $k$th parameter, then the upper bound $U_ {b_k}$ and the lower bound $L_{b_k}$ of parameter ${b_k}$ can be solved by the standard SDP.
    \begin{itemize}
	\item For the lower bound problem:
	\begin{equation}\label{SDP_l}
		\begin{array}{rl}\min\limits& {b_k}\\ \text{subject to}& - 1 \le {b_k}, \cdots ,{b_{{4^n}}} \le 1\\ &Q\succeq0\end{array}
	\end{equation} 
    \end{itemize}
    \begin{itemize}
        \item For the upper bound problem:
	\begin{equation}\label{SDP_u}
		\begin{array}{rl}\min\limits& {-b_k}\\ \text{subject to}& - 1 \le {b_k}, \cdots ,{b_{{4^n}}} \le 1\\ &Q\succeq0\end{array}
	\end{equation}
    \end{itemize}
\end{proposition}

\begin{remark}
    Natural index order refers to the sequence of natural numbers with indexes starting from 1. Corresponds to the index of the Bloch components, from $b_1$ to $b_{4n}$.
\end{remark}

\begin{remark}
     Where $Q$ is the density matrix representation transformed by the Bloch parameter representation  \cite{gamel2016entangled}. It is formally expressed as:
\begin{equation}
    Q=\frac{1}{2^n}\sum\limits_{j = 1}^{{4^n}} {{b_j}{T_j}}
\end{equation}
$T_j$ is the local observables used to obtain $b_j$, which is obtained by tensor product of the Pauli operators ${\sigma _i},i \in \left\{ {0,1,2,3} \right\}$. Detailed description is in the remark to Appendix \ref{Appendix D}.
\end{remark}
   
\begin{proof}
    The proof of the proposition is in Appendix \ref{Appendix D}.
\end{proof}

Building on the semidefinite programming (SDP) framework formalized~\cite{vandenberghe1996semidefinite}, we numerically validate Theorem~\ref{the upper 2-qubit} and hypothesis~\ref{the upper n-qubit}. The convexity of the SDP guarantees global optimality of the solutions, thereby ensuring the statistical reliability of our numerical results.


It is worth emphasizing that optimization based solely on the pure state space structure is non-convex and challenging to verify numerically. The numerical simulation framework is formalized under a predefined "pure + mixed state" space structure. Theoretically, the upper bound of posterior information obtained through SDP formalization under this hybrid structure is guaranteed to be at least not less than that under the "pure state" formulation. When simulation results from the hybrid approach support the hypothesis, this convergence can serve as empirical evidence for hypothesis validation.

We conducted extensive numerical simulations to estimate the information content upper bound of the quantum system. Three state groups were investigated: (i) 100 randomly generated pure states, (ii) 50 random tensor-product states, and (iii) 50 randomly rotated Bell states. Their respective information contents were computationally quantified. The results of the numerical experiment on 2-qubit are shown in Fig.~\ref{Fig.1}.

\begin{figure*}[htbp!]
	\subfloat[Information content of random pure states]{
	\includegraphics[width=0.3\columnwidth]{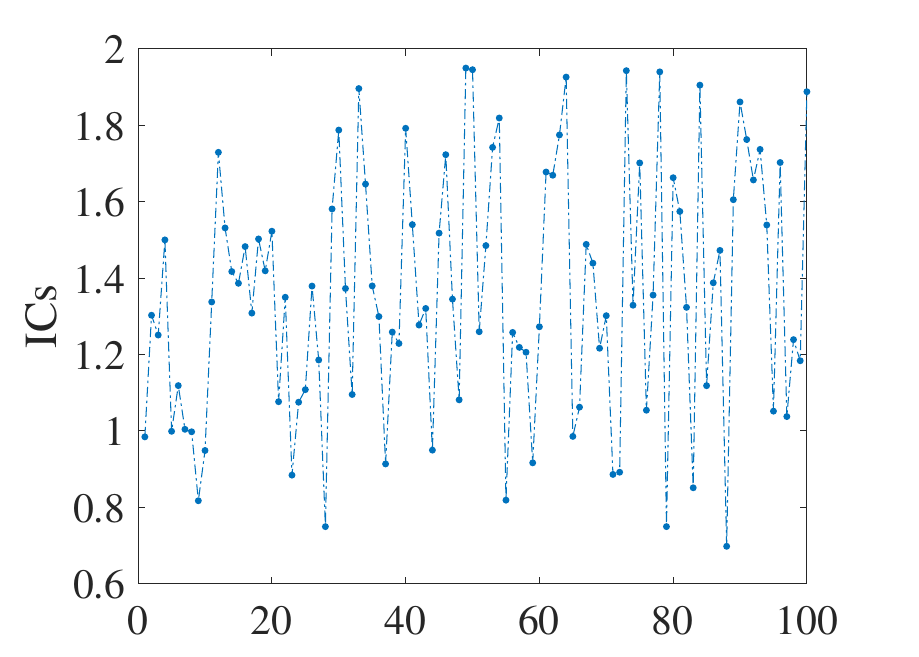}            \label{Fig.1(a)}
}
	\subfloat[Information content of random tensor product state]{
	\includegraphics[width=0.3\columnwidth]{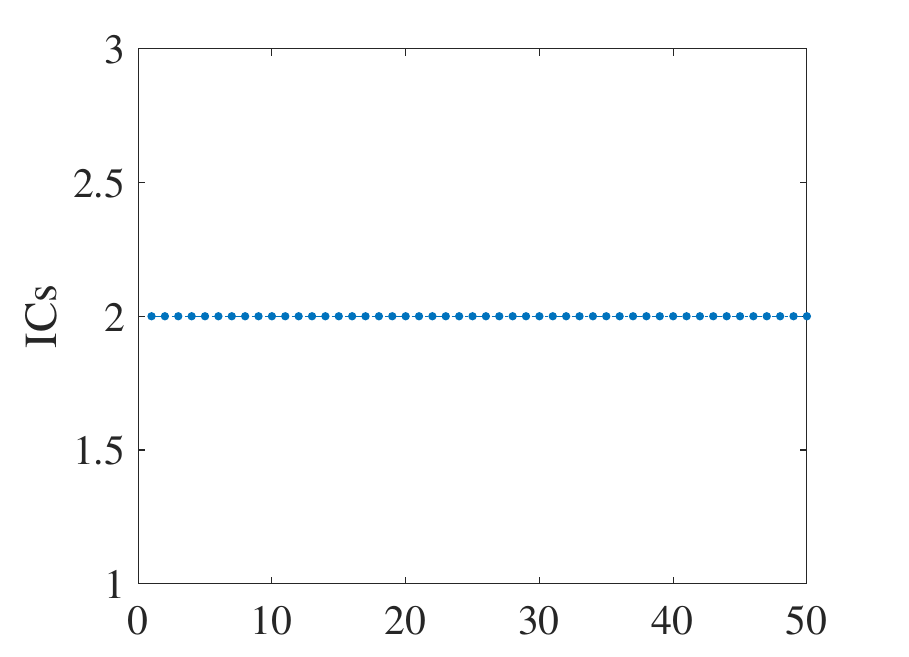}             \label{Fig.1(b)}
}
	\subfloat[Information content of random Bell states]{
		\includegraphics[width=0.3\columnwidth]{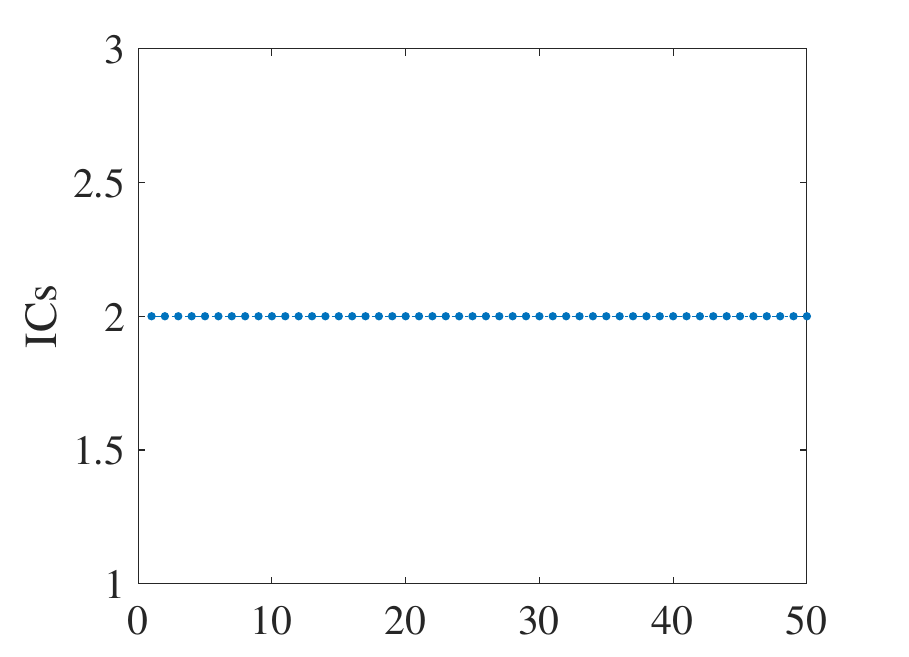} \label{Fig.1(c)}
	}	
\caption{The information content contained in various types of quantum states in a 2-qubit system. }\label{Fig.1}
\end {figure*}

\begin{figure*}[htbp!]
	\subfloat[Information content of random pure states]{
	\includegraphics[width=0.3\columnwidth]{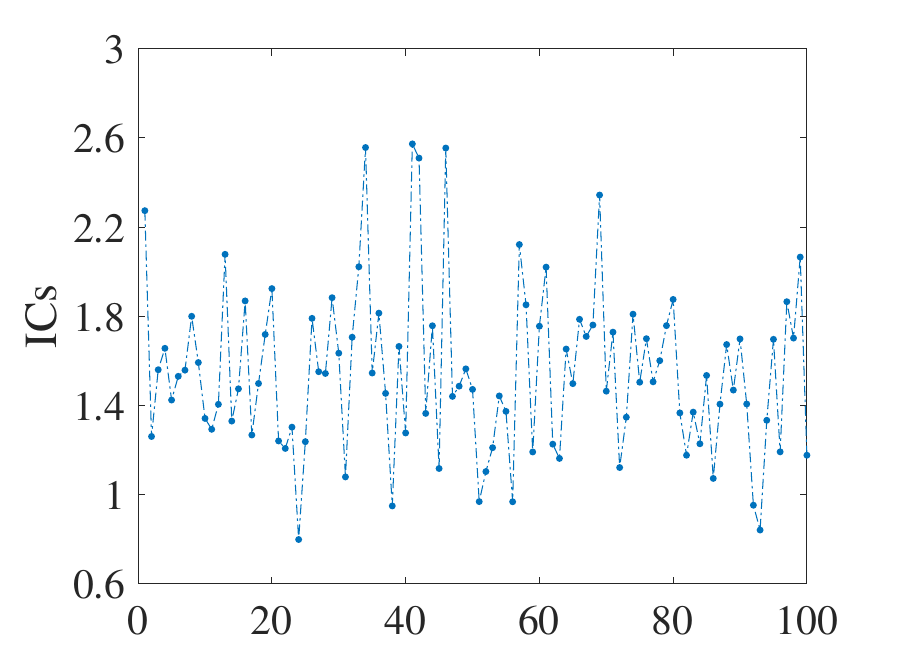} \label{Fig.2(a)}
}
	\subfloat[Information content of random tensor product states]{
	\includegraphics[width=0.3\columnwidth]{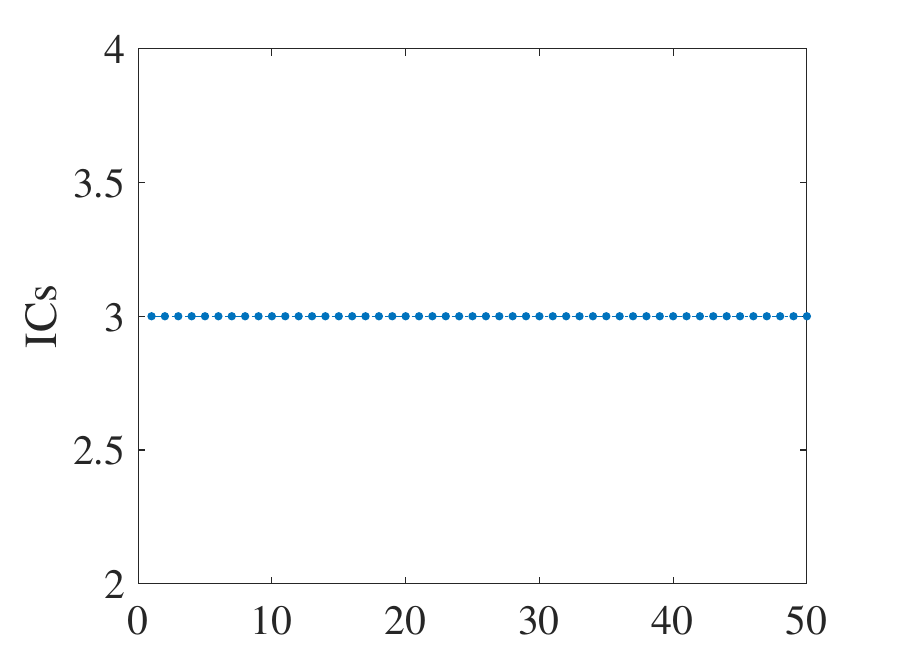} \label{Fig.2(b)}
}
	\subfloat[Information content of random GHZ states]{
		\includegraphics[width=0.3\columnwidth]{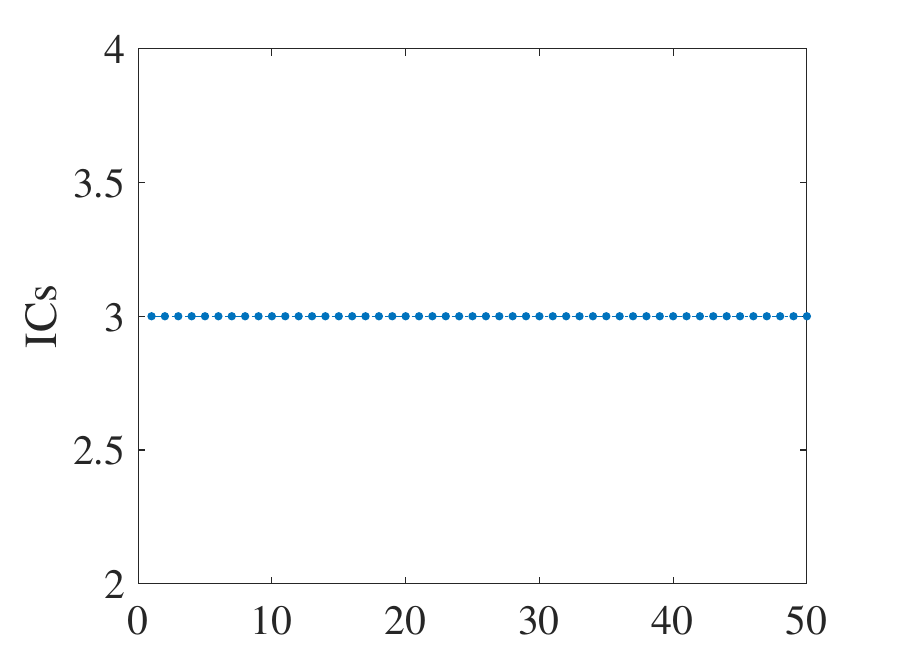} \label{Fig.2(c)}
	}	
\caption{The information content contained in various types of quantum states in a 3-qubit system. }\label{Fig.2}
\end {figure*}

Numerical experiments calculate the information content according to natural index order. In Fig.\ref{Fig.1}, we can see that any state of different categories is consistent with the conclusion in Theorem \ref{the upper 2-qubit}. In particular, in Fig.\ref{Fig.1(b)} we can see that when the system is a tensor product state, its information content is exactly equal to 2. At this time, the calculation result corresponds to the situation in the lemma \ref{mu=1}. Bell states serve as a typical example of quantum entanglement. We randomly select one of the four Bell states for calculation and the result is shown in Fig.\ref{Fig.1(c)}, and its information amount is exactly equal to 2. This corresponds to the situation in Lemma \ref{mu=0}. 

In the case of the 3-qubit quantum system, we also selected 100 random pure states, 50 random tensor product states and 50 random rotated three-valued Greenberger-Horne-Zeilinger (GHZ) states. The results of the numerical experiment are shown in Fig.\ref{Fig.2}.

Comparing the case of 2-qubit systems, we obtained similar conclusions. In other words, the total amount of information contained in random quantum pure states is no more than 3 (that is shown in Fig.\ref{Fig.2(a)}). When the selected state is a tensor product state or a 3-qubit GHZ state, we can always accurately obtain that the system information amount is exactly equal to 3 (shown in Fig.\ref{Fig.2(b)} and \ref{Fig.2(c)}). 

In fact, we have done many sets of numerical experiments, all of which have given the same conclusion. According to the Hypothesis, we found that regardless of the number of quantum bits $n$, the process of calculating the information content using our method can always be represented as an SDP that converges to $n$. Based on the conclusion given by the proposition \ref{SDP}, we generalized the Hypothesis \ref{the upper 2-qubit} to $n$-qubit system.

We believe that Hypothesis \ref{the upper n-qubit} not only clarifies the consistency between the quantum world and the classical world, but is also essential for the related development of quantum computing. Specifically, Shannon's noiseless coding theorem \cite{cover1999elements} states that the upper bound of the expected coding length $L_*$ of the optimal coding scheme for any given information source (probability distribution $D$) satisfies
\begin{equation}
     L_*\leq S(D)+1 .
\end{equation}
where $S(D)$ is the Shannon entropy of $D$. 

The upper bound hypothesis of posterior information proposed in this paper also provides a similar quantitative bound: \textbf{In order to determine the state of any unknown $n$-qubit system, the upper bound of the posterior information required by optimal measurement process is $n$ classic bits.} 

\section{\label{conclusion}Conclusion and outlook}

We have analyzed the posterior information and provided the information content contained in the parameters of the quantum state from an information divergence perspective. Based on the structural constraints of the quantum state, we obtained the upper bound of the posterior information content for the 2-qubit system, which equals 2. The SDP-based numerical analysis enables the extension of our findings to arbitrary $n$-qubit systems, from which we formulated the hypothesis that the upper bound on the posterior information content for an $n$-qubit system is $n$ (Hypothesis~\ref{the upper n-qubit}).

This work provides important theoretical evidence for revealing an elegant consistency between classical and quantum systems in the information-theoretic framework. If one defines the posterior information of a classical-bit system in the same way as in the quantum case, the posterior information content of a $n$-classical-bit system is also exactly equal to $n$. The posterior information upper bound hypothesis (Hypothesis \ref{the upper n-qubit}) explicitly shows that the composite system constructed by quantum theory can ensure that the information requirements will not lose control due to the exponentially increase of parameters. The composite quantum system as an information carrier can still be constrained by basic information processing principles. 

Existing upper bounds on quantum state information (such as the Holevo bound~\cite{holevo1973bounds} and \v{C}aslav Brukner et al.~\cite{PhysRevLett.83.3354}) are all obtained based on specific measurement setups. On the other hand, both the parameter scale and the ontic embedding are intrinsic properties of quantum states~\cite{pusey2012reality}. These works seem to suggest that: although the intrinsic information content of quantum states is likely $super-O(n)$ in scale, the quantum measurement process fundamentally limits our ability to acquire the internal information of quantum states, such that the information that can be leaked by measurement results is only $O(n)$. Our analysis indicates that, due to the high degree of correlation among parameters, the intrinsic information content of quantum states is also $O(n)$. The same explanation may also apply to the ontic embedding, though we have not explicitly presented a mathematical model for this explanation in the paper.

Especially in the realm of quantum state tomography, we noticed that Hypothesis \ref{the upper n-qubit} has very striking connection with the work done by Huang et al., which related the classical and quantum worlds\cite{huang2020predicting}. They proposed an efficient method by using "classical shadow", which shows that no matter how the dimension of the quantum state is, the $M$ scalar parameters defined on the quantum state can be accurately predicted with a high success rate using only $O(\log M)$ measurements. In other words, for a quantum state of $n$ qubits which is characterized by $M=2^n$ parameters, if one wants to estimate this state using diverse methods, at least $O(n)$ measurements are required. Our posterior information upper bound hypothesis explains Huang's result in terms of the information content of quantum states: An $n$-dimensional quantum system contains at most $n$ classical bits of the posterior information. As long as each measurement can obtain $O(1)$ classical information in a statistical sense, then only $O(n)$ measurements can be used to arbitrarily approximate the measured quantum state. In summary, the posterior information upper bound hypothesis gives a asymptotic lower bound on the minimum number of local measurements required to estimate the state of a quantum system. Similar to the Shannon limit~\cite{shannon1948mathematical} that establishes the upper bound on the communication efficiency of a noisy channel, our posterior information upper bound establishes the asymptotic limit on the efficiency of generalized quantum tomography.

\section*{Acknowledgements}
This work is funded in part by the National Natural Science Foundation of China (61876129) and the European Union's Horizon 2020 research and innovation programme under the Marie Skodowska-Curie grant agreement No. 721321.

\bibliographystyle{quantum}
\bibliography{quantum}

\onecolumn

\appendix

\section{\label{Appendix E}The general form of a quantum pure state in Bloch representation}

The Bloch formalization is equivalent to the density operator formalization, and its advantage is that it provides a more geometrically intuitive representation in low-dimensional cases and can more conveniently represent and analyze quantum correlations. The definition of the Bloch representation relies on the observables derived from a specific fiducial measurement. Typically, for a quantum bit, the following three Pauli observables (matrices) can be used: 
\begin{equation}    \sigma_1=\begin{pmatrix}0&1\\1&0\end{pmatrix}\quad\sigma_2=\begin{pmatrix}0&-i\\i&0\end{pmatrix}\quad\sigma_3=\begin{pmatrix}1&0\\0&-1\end{pmatrix}
\end{equation}
For any 2-qubit composite system, if its density operator is expressed as $\rho$, then its Bloch representation can be expressed as $\ \psi =(\boldsymbol{\alpha},\boldsymbol{\beta},C)$. Among them, $\boldsymbol{\alpha} $ and $\boldsymbol{\beta}$ are 
\begin{equation}
    \boldsymbol{\alpha}=\begin{pmatrix}\operatorname{tr}((\sigma_1\otimes I)\rho)\\\operatorname{tr}((\sigma_2\otimes I)\rho)\\\operatorname{tr}((\sigma_3\otimes I)\rho)\end{pmatrix}, \boldsymbol{\beta}=\begin{pmatrix}\operatorname{tr}((I\otimes\sigma_1)\rho)\\\operatorname{tr}((I\otimes\sigma_2)\rho)\\\operatorname{tr}((I\otimes\sigma_3)\rho)\end{pmatrix}
\end{equation}

and $C$ is a 3 × 3 correlation matrix defined as
\begin{equation}
    C=\begin{pmatrix}\operatorname{tr}((\sigma_1\otimes\sigma_1)\rho)&\operatorname{tr}((\sigma_1\otimes\sigma_2)\rho)&\operatorname{tr}((\sigma_1\otimes\sigma_3)\rho)\\\operatorname{tr}((\sigma_2\otimes\sigma_1)\rho)&\operatorname{tr}((\sigma_2\otimes\sigma_2)\rho)&\operatorname{tr}((\sigma_2\otimes\sigma_3)\rho)\\\operatorname{tr}((\sigma_3\otimes\sigma_1)\rho)&\operatorname{tr}((\sigma_3\otimes\sigma_2)\rho)&\operatorname{tr}((\sigma_3\otimes\sigma_3)\rho)\end{pmatrix}
\end{equation}
In Bloch representation, the elements in $\alpha$ are usually defined as $b_2-b_4$, the elements in $\beta$ are defined as $b_5-b_7$, and the elements in $C$ are defined as $b_8-b_{16}$ in row-major order. In addition, $b_1 = \operatorname{tr}((I\otimes I)\rho)=1 $.

For the detailed proof of Eq.\ref{2-qubit_1}-\ref{2-qubit_3} in the main text, please refer to the Ref.\cite{gamel2016entangled}. We only give a general idea. There are unitary transformations between any pure states that allow them to be converted to each other, which means that all pure states can be generated through unitary transformations based on direct product states. In addition, local transformations do not change the form of the diagonal matrix of $C$ after SVD decomposition. This diagonal matrix can only be controlled by non-local transformations. Non-local transformations have the following general form:
\begin{equation}
    \tilde{U}(\theta_1,\theta_2,\theta_3)=\tilde{U}_1(\theta_1)\tilde{U}_2(\theta_2)\tilde{U}_3(\theta_3)
\end{equation}
where
\begin{equation}
    \tilde{U}_j(\theta_j)=\exp\left(\frac{\theta_j}{2}\sigma_j\otimes\sigma_j\right)
\end{equation}
 
Acting on the pure state of the direct product of $\boldsymbol{\alpha}=(1,0,0)^{\mathrm{T}}$ and $\boldsymbol{\beta}=(1,0,0)^{\mathrm{T}}$, we can obtain

\begin{equation}
    \label{ABC}\boldsymbol{\alpha}=\begin{pmatrix}\cos\theta\\0\\0\end{pmatrix},\boldsymbol{\beta}=\begin{pmatrix}\cos\theta\\0\\0\end{pmatrix},C=\left(\begin{array}{ccc}1&0&0\\0&0&\sin\theta\\0&\sin\theta&0\end{array}\right )
\end{equation}
where $\theta=\theta_2-\theta_3$. 

Note that for the sake of discussion, the form of the pure state given by Eq.\ref{2-qubit_1}-\ref{2-qubit_3} is not consistent with the form given by Eq.\ref{ABC}. In fact, as long as the parameters in $O_1$ and $O_2$ are represented as:
\begin{equation}\sin\omega=-1,\cos\varphi=1,\cos k=1\end{equation}
It gets the same form as in Eq.\ref{2-qubit_1}-\ref{2-qubit_3}.




On this basis, we can calculate that:

\begin{equation}
    \boldsymbol{\alpha}=\left(\begin{array}{c}\cos\theta\cos\phi\cos k\\-\cos\theta\cos\phi\sin k\\\cos\theta\sin\phi\\\end{array}\right)\\,\boldsymbol{\beta}=\left(\begin{array}{c}\cos\theta\cos\phi'\cos k'\\-\cos\theta\cos\phi'\sin k'\\\cos\theta\sin\phi'\end{array}\right)
\end{equation}
In particular, when $\cos \theta =0$, $\boldsymbol{a},\boldsymbol{b}=\boldsymbol{0}$. The parameters of the matrix C can be represented as:

    \begin{equation}
        \begin{split}
            c_{11} &= \cos \phi \cos k \cos \phi' \cos k' \\
            &\quad + \sin \theta \left( \cos \omega \sin k + \sin \omega \sin \phi \cos k \right) 
            \left( \cos \omega' \sin k' + \sin \omega' \sin \phi' \cos k' \right) \\
            &\quad- \sin \theta \left( \sin \omega \sin k - \cos \omega \sin \phi \cos k \right) 
            \left( \sin \omega' \sin k' - \cos \omega' \sin \phi' \cos k' \right) \\
            &= \cos \phi \cos k \cos \phi' \cos k' \\
            &\quad- \sin \theta \sqrt{1 - \cos^2 \phi \cos^2 k} 
            \sqrt{1 - \cos^2 \phi' \cos^2 k'} 
            \cos(\omega + \omega' + \omega_1 + \omega_1')
        \end{split}
    \end{equation}

    \begin{equation}
    \begin{split}
        c_{12} &= -\cos \phi \cos k \cos \phi' \sin k' \\
        &\quad + \sin \theta \left( \cos \omega \sin k + \sin \omega \sin \phi \cos k \right) 
        \left( \cos \omega' \cos k' - \sin \omega' \sin \phi' \sin k' \right) \\
        &\quad - \sin \theta \left( \sin \omega \sin k - \cos \omega \sin \phi \cos k \right) 
        \left( \sin \omega' \cos k' + \cos \omega' \sin \phi' \sin k' \right) \\
        &= -\cos \phi \cos k \cos \phi' \sin k' \\
        &\quad + \sin \theta \sqrt{1 - \cos^2 \phi \cos^2 k} 
        \sqrt{1 - \cos^2 \phi' \sin^2 k'} 
        \sin(\omega + \omega' + \omega_1 + \omega_2')
    \end{split}
\end{equation}
where
\begin{equation}
   \resizebox{.98\hsize}{!}{$\begin{aligned}&\sin\omega_1=\frac{\sin k}{\sqrt{\sin^2k+\sin^2\phi\cos^2k}},\sin\omega_1^{'}=\frac{\sin k^{'}}{\sqrt{\sin^2k+\sin^2\phi^{'}\cos^2k^{'}}},\sin\omega_2^{'}=\frac{\cos k^{'}}{\sqrt{\cos^2k^{'}+\sin^2\phi^{'}\sin^2k^{'}}}\\&\cos\omega_1=\frac{\sin\phi\text{cos}k}{\sqrt{\sin^2k+\sin^2\phi\text{cos}^2k}},\cos\omega_1^{'}=\frac{\sin\phi^{'}\cos k^{'}}{\sqrt{\sin^2k^{'}+\sin^2\phi^{'}\cos^2k^{'}}},\cos\omega_2^{'}=\frac{\sin\phi^{'}\sin k^{'}}{\sqrt{\cos^2k^{'}+\sin^2\phi^{'}\sin^2k^{'}}}\end{aligned}$} 
\end{equation}

\section{\label{Appendix A}The prior information content of 2-qubit pure state}

In order to clarify the information signification of pure state 2-quantum system, we will give a specific definition of prior information content. For the pure state 2-qubit Bloch vector represented by Eq.\ref{psi}, its prior $\chi^2$ divergence information is defined as:
\begin{equation}\label{proi info}
   I_{\chi^{2},prior}(\psi)=\sum_{i=2}^{16}D_{\chi^{2}}(\psi_{i},0)
\end{equation}
where 0 is reference distribution of the prior $\chi^2$ divergence information content \cite{Masanes_2011}.

\begin{remark}
    Each parameter is estimated using a trivial prior (equivalent to maximum likelihood estimation (MLE). The calculation of information also uses the median point of the trivial prior as the reference distribution. Here, the trivial prior refers to the uniform distribution on the value interval $[-1,1]$ of the Bloch parameter.
\end{remark}

The asymptotic upper bound of the prior information of an $n$-qubit system is equal to the square of the modulus of the $n$-qubit pure-state Bloch representation, where asymptotic means that the number of samples tends to infinity. If the asymptotic requirement is removed, the maximum prior information of an $n$-qubit system can be $4^n$ (corresponding to the case where the number of samples is equal to 1).

As a comparison, we paraphrase the quantum information measure defined by Caslav Brukner and Anton Zeilinger\cite{PhysRevLett.83.3354}. Consider a static experimental arrangement with $n$ possible outcomes, with known probability $\vec{p}=(p_{1},...,p_{j},...,p_{n})$. 
Obviously, since $\sum_ip_i=1$, not all vectors in the probability space are possible. In fact, this gives the minimum length of $\vec{p}$ when all probability are equal $(p_i = 1/n)$. In this case, the information measure is defined as:
\begin{equation}
    I(\vec{p})=\mathcal{N}\sum_{i=1}^n\left(p_i-\frac{1}{n}\right)^2.
\end{equation}
When encoding a maximum of $k$ bits of information, that is, $n = 2^k$. The normalization is $\mathcal{N}=2^{k}k/(2^{k}-1)$. As a representation of prior information content, it is clear that $(p_i = 1/n)$ corresponds to the case where the reference distribution in Eq.\ref{proi info} is equal to 0. Therefore, over the complete set of $m$ complementary experiments \cite{bohr2010atomic}, the total information content of the system can be defined as:
\begin{equation}
    I_{total}=\sum_{j=1}^{m}I_{j}(\vec{p})
\end{equation}

It needs to be emphasized that the difference between prior information content and posterior information content is that the dependencies between the parameters of a quantum state are not considered when calculation the prior information content. In other words, even if the values of the first-$k$ Bloch parameters of a quantum state have been identified by measurements, those values will not be used to determine the value range and information content of the $k+1th$ parameter (i.e., we assume a prior that the value range of all parameters belong to [-1,1]). In contrast, the calculation of the posterior information content needs to take into account the dependencies between parameters.

\section{\label{Appendix B}Pythagorean properties of $\chi^2$ divergence}

\begin{lemma}[Pythagorean properties]
    Assume that $t_1$,$t_2$,$t_3$ are components of the Bloch vector, and $t_1^2 = t_2^2 + t_3^2$. then:
\begin{equation}
    D_{X^2}(t_1,0)\equiv D_{X^2}(t_2,0)+D_{X^2}(t_3,0)
\end{equation}
\end{lemma}

\begin{proof}
    From the relationship between probability and Bloch parameters, It can be given:
    \begin{equation}
        \begin{aligned}
D_{X^{2}}(t_{1},0)& =\frac{(\frac{t_{1}+1}{2})^{2}}{\frac{1}{2}}+\frac{(\frac{1-t_{1}}{2})^{2}}{\frac{1}{2}}-1  \\
&=\frac{(t_{1}+1)^{2}}{2}+\frac{(1-t_{1})^{2}}{2}-1 \\
&=\frac{t_{1}^{2}+2t_{1}+1}{2}+\frac{t_{1}^{2}-2t_{1}+1}{2}-1 \\
&=t_{1}^{2}
\end{aligned}
    \end{equation}
    In the same way, the above calculation can be performed for $t_2$ and $t_3$, and the lemma is proved.  
\end{proof}

\section{\label{Appendix C}The unique determination of pure states in 2-qubit systems}
The Bloch representation of a quantum pure state involves many parameters and they are coupled to each other. This is not conducive to analyze information content. First notice, $\omega+\omega'$ can be discussed as a whole parameter without discussing the values of $\omega$ and $\omega'$ separately. then, We derive the following lemma.

\begin{lemma}[Unique determination]
For a Bloch vector of pure state in 2-qubit system $\psi=(\boldsymbol{\alpha}, \boldsymbol{\beta}, \boldsymbol{C})$ with $\cos^2\phi\cos^2k\neq0$ or $\cos^2\phi'\cos^2k'\neq0$, when $\boldsymbol{\alpha}$, $\boldsymbol{\beta}$ and the first two components $c_{11}$, $c_{12}$ in the matrix $\boldsymbol{C}$ are determined, then the remaining components of this Bloch vector are uniquely determined.
\end{lemma}

\begin{proof}
    Given $\boldsymbol{\alpha}$ and $\boldsymbol{\beta}$ ($\cos{\theta}\geq0$), then $\sin{\phi}$ and $\sin{\phi}'$ can be uniquely determined, but the signs of the remaining trigonometric function values in $\boldsymbol{\alpha}$ and $\boldsymbol{\beta}$ cannot be uniquely determined. We can assume that $\cos\phi,\cos k,\sin k$ and $\cos\phi',\cos k',\sin k'$ satisfy the constant value requirements of $\boldsymbol{\alpha}$ and $\boldsymbol{\beta}$. Then $-\cos\phi,-\cos k,-\sin k$ and $-\cos\phi',-\cos k',-\sin k'$ also satisfy requirements of $\boldsymbol{\alpha}$ and $\boldsymbol{\beta}$. In other words, if $\phi,k$ and $\phi',k'$ satisfy the requirements, $\pi-\phi,\pi+k$ and $\pi-\phi',\pi+k'$ also satisfy the requirements without considering the period. Now, there are four sets of satisfying angles: $(\phi,k,\phi',k')$, $(\pi-\phi,\pi+k,\phi',k'),(\phi,k,\pi-\phi',\pi+k')$, and $(\pi-\phi,\pi+k,\pi-\phi',\pi+k')$.

    Next, we analyse the values of $c_{11}$ and $c_{12}$. Assume that the current value of $\omega+\omega'$ makes $c_{11}$ and $c_{12}$ the aforementioned constant values. According to the above analysis, the current values of the angle parameters are changed from $(\phi,k,\phi',k')$ to $(\pi-\phi,\pi+k,\phi',k')$ or $(\phi,k,\pi-\phi',\pi+k')$, the values of $\boldsymbol{\alpha}$ and $\boldsymbol{\beta}$ remain unchanged. Correspondingly, when the above changes occur, the value of $\omega_1+\omega_1'(\omega_1+\omega_2')$ will become $\omega_1+\omega_1'+\pi(\omega_1+\omega_2'+\pi)$ within one cycle. At this time, the current values of $\omega$ and $\omega'$ will no longer be able to keep $c_{11}$ and $c_{12}$ at constant values, and $\omega+\omega'$ can only be modified to $\omega+\omega'+\pi$.
    By the same token, when $(\phi,k,\phi',k')$ changes into $(\pi-\phi,\pi+k,\pi-\phi',\pi+k')$, to keep $c_{11}$ and $c_{12}$ as constant values, $\omega+\omega'$ should be modified to $\omega+\omega'+2\pi$. 
    
    In summary, the same values of $\boldsymbol{\alpha}$, $\boldsymbol{\beta}$, $c_{11}$ and $c_{12}$ can be obtained at the same time for only the following four sets of rotated matrix angle with ignoring the cycles: 
    \begin{equation}  
        \begin{aligned}&\left(\mathrm{a}\right)\phi,k,\phi',k,\omega+\omega'\\&\left(\mathrm{b}\right)\pi-\phi,\pi+k,\phi',k',\omega+\omega'+\pi\\&\left(\mathrm{c}\right)\phi,k,\pi-\phi',\pi+k',\omega+\omega'+\pi\\&\left(\mathrm{d}\right)\pi-\phi,\pi+k,\pi-\phi',\pi+k',\omega+\omega'+2\pi
        \end{aligned}
    \end{equation}
    By calculating $\boldsymbol{C}$ through these four sets of values, it can be seen that $\boldsymbol{C}$ will not change. In other words, it can be considered that the trigonometric function values of $\phi$, $k$, $\phi'$, $k'$ (including their signs) are determined(i.e. quantum state is uniquely determined), after given $\boldsymbol{\alpha}$, $\boldsymbol{\beta}$, $c_{11}$ and $c_{12}$.

\end{proof}
    Therefore, when $\cos^2\phi\cos^2k\neq0$ or $\cos^2\phi'\cos^2k'\neq0$, to calculate the information content of a quantum state under the constraints of the pure state structure, only need to calculate the information contained in $\boldsymbol{\alpha}$, $\boldsymbol{\beta}$, $c_{11}$ and $c_{12}$. More detailed analysis can be found in \cite{zhang2022application}.

\section{\label{Appendix D}The proof of proposition \ref{SDP}}

We redescribe Proposition.1 in Section IV.
\begin{proposition}
    For any positive integer $n$, any (standard) $n$-qubit system, its first $k-1$ parameters $b_i,\;1 \le i \le k - 1 \le {4^n}$ are given in the natural index order of the standard Bloch representation. Assume ${b_k}$ is the $k$-th parameter, then the upper bound $U_ {b_k}$ and the lower bound $L_{b_k}$ of parameter ${b_k}$ can be solved by the standard SDP.
    \begin{itemize}
	\item For the lower bound problem:
	\begin{equation}\label{sdp_l}
		\begin{array}{rl}\min\limits& {b_k}\\ \text{subject to}& - 1 \le {b_k}, \cdots ,{b_{{4^n}}} \le 1\\ &Q\succeq0\end{array}
	\end{equation} 
    \end{itemize}
    \begin{itemize}
        \item For the upper bound problem:
	\begin{equation}\label{sdp_u}
		\begin{array}{rl}\min\limits& {-b_k}\\ \text{subject to}& - 1 \le {b_k}, \cdots ,{b_{{4^n}}} \le 1\\ &Q\succeq0\end{array}
	\end{equation}
    \end{itemize}
\end{proposition}

\begin{proof}
    According to the standard form of the SDP problem as follows:
        \begin{equation}\label{DSDP}
		\begin{array}{rl}\min\limits_{X\in S^{2^n}}&\langle C,X\rangle_{\mathbb{S}^{2^n}} \\ \text{subject to}&\langle A_i,X\rangle_{\mathbb {S}^{2^n}} \le d_i,\\ &X\succeq0\end{array}
	\end{equation}
 
    If the above proposition is true, the objective function needs to satisfy $\langle C,Q\rangle=b_k$, and the constraints need to satisfy $-1\le\langle A_i,Q\rangle=b_i \le 1$,$i \in [k,4^n]$. 

    Therefore, it is only necessary to prove that for each $b_k$ corresponds to a matrix $C$ and a set of $A_k$, and they are all Hermitian matrix, and make the formalization of Eq.(\ref{sdp_l}) and (\ref{sdp_u}) can be obtained.

    In a $n$-qubit composite system, the Bloch parameter $b_k$ can be determined from local fiducial measurements. Its measurement process can be expressed as:
    \begin{equation}
        \langle A_{k},Q\rangle = b_k
    \end{equation}
    
    Where $A_k$ is a the measurement operator, and each Bloch parameter corresponds to a measurement operator $A_k$. For example, taking the seventh parameter $b_7$ in the 2-qubit quantum system as an example, the measurement operator $A_8$ corresponding to the seventh parameter $b_8$ is ${\sigma _1} \otimes {\sigma _1}$, and its matrix is expressed as
    \begin{equation}\label{b7}
	{\sigma _1} \otimes {\sigma _1} = \left[ {\begin{array}{*{20}{c}}
			0&0&0&1\\
			0&0&1&0\\
			0&1&0&0\\
			1&0&0&0
	\end{array}} \right]
    \end{equation}

    Because $A_k$ is a tensor matrix composed of ${\sigma _i},i \in \left\{ {0,1,2,3} \right\}$, it must be a Hermitian matrix.
    
    When the $k$-th Bloch parameter $b_k$ is calculated, it is obvious that the matrix $C$ should be equal to the measurement operator $A_k$ of $b_k$.
    Therefore, there is always a set of identical matrix $C$ and $A_k$ corresponding to each Bloch parameter $b_k$, making the process of calculating the upper and lower bounds of the posterior information (Eq.(\ref{sdp_l}) and (\ref{sdp_u})) into a standard SDP.
    \end{proof}
    The above proposition shows that even in a high-dimensional qubit space $(n>2)$, the process of calculating the upper bound of the posterior information amount can always be analogized to the situation of solving the 2-qubit space, and it can be summarized as an SDP problem.

\section{Supplemental Material: The justifications for defining information content through $\chi^2$ divergence}

\subsection{introduction}
This section considers how to formally define information and explains the justifications for defining information based on the $\chi^2$ divergence. Consider an event $E$ that may occur in a random experiment. When we actually observe event $E$, how much information do we obtain? Obviously, the degree of surprise of event $E$ depends on the probability of $E$ occurring. The more unlikely the event, the greater the degree of surprise, and the more information it can offer. That is, "Information means surprise"\cite{6773024}. 
\subsection{How to quantify information}

"Surprise" refers to the psychological shock of a cognitive agent. It is caused by the change in cognitive state of the agent. Generally speaking, we describe the cognitive state of an agent with probability distribution and suppose that it can be changed by the occurrence of random events. In this way, we can further quantify ”surprise” in terms of the estimate likelihood between a sampling distribution and an underlying (real) distribution, both of which describes a cognitive state, respectively. Specifically, assuming that the cognitive state before a random event occurs is described by probability distribution $q$, and the cognitive state after the random event occurs is described by probability distribution $p$, then the ”surprise” measure caused by the occurrence of the random event (that is, the amount of information brought by the occurrence of the random event) is recorded as $d(p,q)$. \textbf{If information implies surprise, then the value of $d(p,q)$ should be negatively correlated with the likelihood that $q$ was misestimated as $p$ due to sampling bias.}

\subsection{Defining information content through $\chi^2$ divergence}
The main text presents a framework for analyzing the upper bound of quantum systems' posterior information. Based on this framework, we consider the definition of information based on the Bernoulli distribution (i.e.,0-1 distribution, 2-category distribution). Assuming that the occurrence of a random event causes the 0-1 distribution corresponding to the system state to change from $q$ to $p$, information content $d(p,q)$ brought by the random event should satisfy the following axiomatic conditions \textbf{Axiom A:}:

\begin{itemize}
    \item \textbf{A1} \textbf{Non-negativity (NN):} Information content brought by the occurrence of any random event should be non-negative.
    
    \item \textbf{A2} \textbf{Likelihood Consistency (LC):} For any two 0-1 distributions $p=[p_1,p_2]$ and $q=[q_1,q_2]$, under any given sample size $m$ (compatible with $p,q$), let the likelihood when the true distribution is $q$ and the sampling distribution is $p$ is $l_{p2q}$, and the likelihood that the true distribution is $p$ and the sampling distribution is $q$ is $l_{q2p}$. Without loss of generality, if $l_{p2q}>l_{q2p}$, then likelihood consistency requires $d(p,q)<d(q,p)$.

    \begin{remark}
        Obviously, the premise for LC to be well defined is that the relationship between $l_{p2q}$ and $l_{q2p}$ is independent of the sample size $m$. The simulation results within the range allowed by numerical accuracy show that the relationship between $l_{p2q}$ and $l_{q2p}$ is indeed independent of the sample size. See the first two subfigures in Fig.1 a) and b). This supplemental material briefly proves that $\chi^2$-divergence complies with likelihood consistency, please refer to Appendix~\ref{lc-CSD}. 
    \end{remark}

    \item \textbf{A3} \textbf{Asymmetry (AS):} $d(p,q)$ is generally not equal to $d(q,p)$.
    \begin{remark}
        Because obviously $l_{p2q}$ is not equal to $l_{q2p}$ in general. A typical example is to consider the estimated likelihood $l_{p2q}$ and $l_{q2p}$ between $p=[0.9,0.1]$ and $q=[1,0]$.
    \end{remark}
    
    \item \textbf{A4} \textbf{Infinitesimal-Scale Likelihood Consistency (ILC):} It is only required that \textbf{A2} holds when the difference between $p$ and $q$ is infinitesimal.
    
    \item \textbf{A5} \textbf{1-Bit}: 
    A binary measurement procedure provides at most one classical bit of (posterior) information.
    \begin{remark}
      This implies that no matter how much prior knowledge we have, the measurement process for determining a two-valued state can yield at most 1 bit of (posterior) classical information. That is, $d([1,0],[1/2,1/2])=1$. 
    \end{remark}
    
    \item \textbf{A6} \textbf{FRD:} The definition of information content should be an approximation of the Fisher-Rao distance (FRD) on an infinitesimal scale. 
    \begin{remark}
        The Fisher-Rao distance can be shown to be the only reasonable distance measure for probability distributions in a certain sense. However, FRD does not satisfy all of Axioms A. In particular, the symmetries of FRD directly conflict with \textbf{Axiom A3 (AS)}. Therefore, the definition of information content should approximate the Fisher-Rao distance (FRD).
    \end{remark}

    \item \textbf{A7} \textbf{Weak Likelihood Consistency (WLC):} Given three 0-1 distributions $p=[p_1=q_1+e,q_2-e]$, $ q=[q_1,q_2]$ and $r=[q_1-e,q_2+e]$, where $e$ is a small positive number, and $m$ is sample size (compatible with $p,q,r$). Let 
    \begin{itemize}
        \item $l_{p2q}$ is the likelihood when the true distribution is $q$ and the sampling distribution is $p$.
        \item $l_{r2q}$ is the likelihood when the true distribution is $q$ and the sampling distribution is $r$.
    \end{itemize}  
    
    WLC requires that, when $l_{p2q}\geq l_{r2q}$ (w.r.t. sample size $m$), then $d(p,q)\leq d(r,q)$. In addition, when $l_{p2q}\leq l_{r2q}$(w.r.t. sample size $m$), then $d(p,q)\geq d(r,q)$. Obviously, if both conditions $d(p,q) \leq d(r,q)$ for some $m$ and $d(p,q)\geq d(r,q)$ for some (different) $m$ are true simultaneously, then WLC requires that $d(p,q) = d(r,q)$.


    \begin{remark}
        Both LC and SLC only consider the consistency between the estimated likelihood and the separation measure $d$ between two probability distributions. A natural extension is to consider the consistency between the estimated likelihood and the separation measure $d$ between more distributions. However, the size relationship between the estimated likelihoods of three (or more) distributions may depend on the sample size $m$. When the magnitude relationship between $l_{p2q}$ and $l_{r2q}$ reverses under different simple sizes of $m$, LC cannot be well defined. Therefore, it is necessary to define LC in a weakened sense as defined by the A7. As an example, consider three distributions: $p = [0.9, 0.1]$, $q = [0.8, 0.2]$, and $r = [0.7, 0.3]$. As shown in Fig.1 c), the relationship between the estimated likelihoods of the three distributions of $p$, $q$, and $r$ may change under different sample sizes. 
    \end{remark}
  
\end{itemize}   

    \begin{figure}[htbp]
        \centering
        \includegraphics[width=0.65\linewidth]{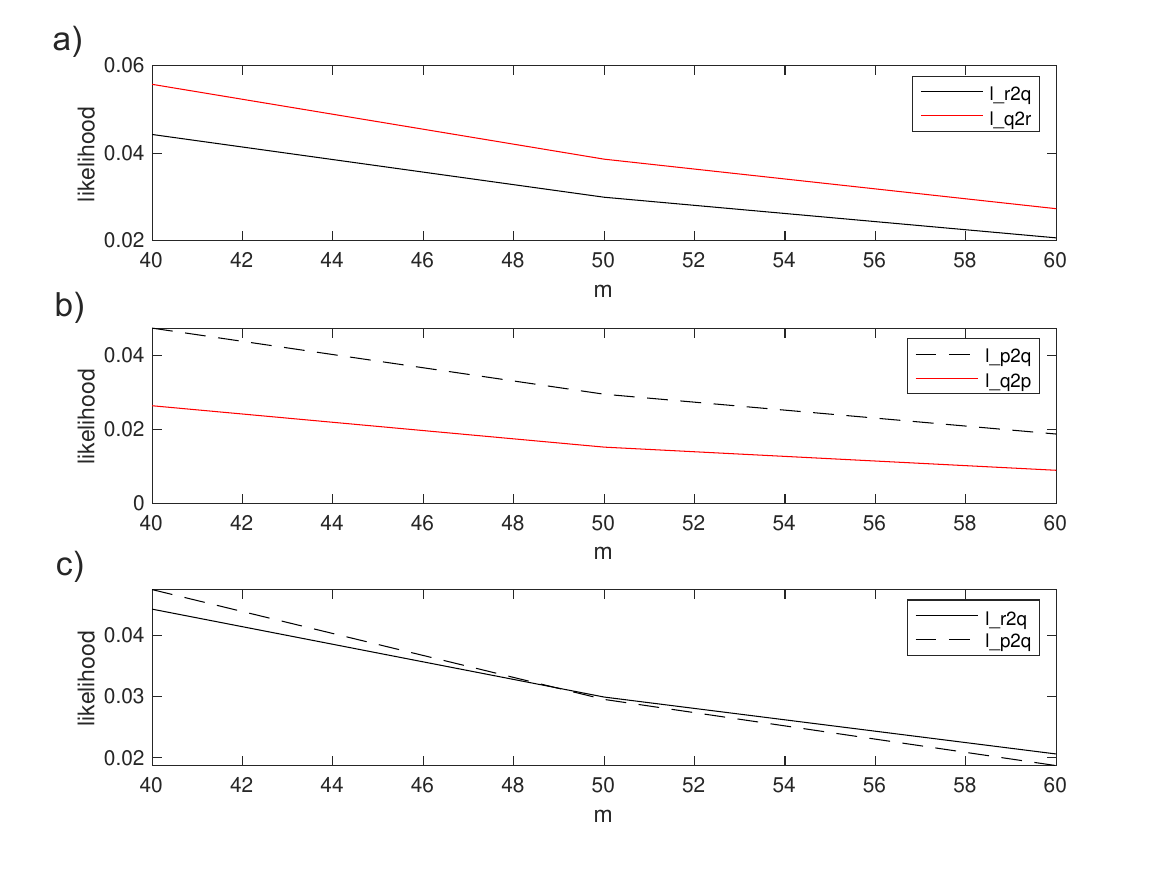}
        \caption{Relationship between the relative size of likelihood and sample size}
        \label{fig:Likelihood}
    \end{figure}

In summary, we mainly consider several candidate definitions $(C)$ of $d(p,q)$.

\begin{itemize}
    \item [C1] \textbf{HD:} $d(p,q) = H(p)-H(q)$, where $H(p)$ is the Shannon entropy of $p$.
    \item [C2] \textbf{FRD:} $d(p,q)=FRD (p,q)$, that is Fisher-Rao Distance.
    \item[C3] \textbf{$\text{KLD}_e$:} $d(p,q)=KLD_e (p,q)$, that is, the KL Divergence with base $e$.
    \item[C4] \textbf{$\text{KLD}_2$:} $d(p,q)=KLD_2 (p,q)$, that is, the KL Divergence with base $2$.
    \item[C5] \textbf{CSD:} $d(p,q)=CSD (p,q)$, that is, the $\chi ^2$ Divergence.
\end{itemize}

The satisfaction of the above candidate $C$ for different premise axioms A is summarized in the following table.

\begin{table}[ht]
\centering
\caption{Satisfaction of Axioms A1-A7 by Different Information Measures}
\label{tab:axiom_satisfaction}
\renewcommand{\arraystretch}{1.5} 
\setlength{\tabcolsep}{12pt} 
\large 
\begin{tabular}{|c|c|c|c|c|c|c|c|}
\hline
 & NN & LC & ILC & AS & 1-Bit & FRD & WLC \\
\hline
\textbf{HD} & $\times$ & $\times$ & $\times$ & $\checkmark$ & $\checkmark$ & $\times$ & $\times$ \\ 
\hline
\textbf{FRD} & $\checkmark$ & $\times$ & $\checkmark$ & $\times$ & $\times$ & $\checkmark$ & $\times$ \\
\hline
\textbf{$\text{KLD}_e$} & $\checkmark$ & $\checkmark$ & $\checkmark$ & $\checkmark$ & $\times$ & $\checkmark^{*}$ & $\times$ \\
\hline
\textbf{$\text{KLD}_2$} & $\checkmark$ & $\checkmark$ & $\checkmark$ & $\checkmark$ & $\checkmark$ & $\checkmark^{*}$ & $\times$ \\
\hline
\textbf{CSD} & $\checkmark$ & $\checkmark$ & $\checkmark$ & $\checkmark$ & $\checkmark$ & $\checkmark^{*}$ & $\checkmark$ \\
\hline
\end{tabular}
\end{table}

\begin{remark}
    Here, $\checkmark^{*}$ means approximately holds. Because $\text{KLD}_e$ is the approximation of FRD at an infinitesimal scale. $\text{KLD}_2$ is the proportional approximation of FRD at an infinitesimal scale, i.e., $\text{KLD}_2 = c * \text{KLD}_e$, where the constant $c$ equals $log_2(exp(1))$. For the proof, please see the Appendix \ref{Relationships}. $\text{KLD}_2$ and $\chi^2$-divergence are (upper bound) approximations of $\text{KLD}_e$\cite{1320776d-9e76-337e-a755-73010b6e4b64,csiszar1967information, AmariNagaoka2000}. Most of the conclusions in the above table can be obtained from the references. In particular,for the proof that $\chi^2$-divergence satisfies WLC, please refer to the Appendix \ref{wlc-CSD}.
\end{remark}

In addition, the formal simplicity of $\chi^2$ divergence (CSD) are very important for analyzing the information upper bound of composite quantum systems. For the distinction between entangled states and separable states, $\chi^2$ divergence can quantify the information upper bound of entanglement resources through behavioral analysis under the framework of local operations and classical communication (LOCC). Most importantly, this paper starts from first principles and strictly analyzes the inherent rationality of $\chi^2$ divergence as a measure of information content according to Axioms A1-A7, which effectively unifies the consistency in quantum and classical scenarios. In summary, the core advantages of $\chi^2$ divergence in the information analysis of composite quantum systems are: the compatibility of its mathematical form with classical statistics, the intrinsic connection to key quantum information metrics, sensitivity to entanglement and non-classical correlations, and computational and experimental efficiency. Combining these justifications, we finally choose $\chi^2$ divergence as the definition of information content.



\subsubsection{The proof of $\chi^2$-divergence satisfying likelihood consistency}
\label{lc-CSD}
\begin{lemma}
    Let ${p}$ and $q$ be any two Bernoulli distributions with parameters in $(0,1)$ (i.e., non-degenerate 0-1 distributions), and let $m\geq1$ be any sample size. Define:
\begin{itemize}
    \item $l_{p2q}$ is the likelihood when the true distribution is ${q}$ and the sampling distribution is ${p}$.
    \item $l_{q2p}$ is the likelihood when the true distribution is ${p}$ and the sampling distribution is ${q}$.
\end{itemize}
Assume that the likelihood is determined by the maximum likelihood estimation rule: given a sample of size $m$~(compatible with $p$ and $q$). Without loss of generality, if $l_{p2q}>l_{q2p}$, then: 
\begin{equation}
    D_{\chi^2}(p,q)<D_{\chi^2}(q,p)
\end{equation}
\end{lemma}

\begin{proof}
    Proof by contradiction: 
    Assume, to the contrary, that according to the definition in lemma 1, without loss of generality, if $l_{p2q}>l_{q2p}$, then: 
    \begin{equation}\label{pq>qp}
    D_{\chi^2}(p,q)>D_{\chi^2}(q,p)
    \end{equation}
    To the 0-1 distribution, we let $p_1=p,~p_2=1-p,~q_1=q,~q_2=1-q$ for simplicity.
    $D_{\chi^2}(p,q)$ can be expressed as:
    \begin{equation}  
    	\begin{split}
    	D_{\chi^2}(p,q)  \equiv  \left( \sum\nolimits_{i} \frac{p^{2}_{i}}{q_{i}} \right) -1 = = \frac{p^{2}}{q}+\frac{(1-p)^{2}}{1-q}-1 =\frac{(p-q)^2}{q(1-q)}
           	\end{split}
    \end{equation}
    In addition, $D_{\chi^2}(q,p)$ can be expressed as:
    \begin{equation}
        D_{\chi^2}(q,p) \equiv\left( \sum\nolimits_{i} \frac{q^{2}_{i}}{p_{i}} \right) -1\frac{(q-p)^2}{p(1-p)}
    \end{equation}
    Obviously, according to the Eq.\ref{pq>qp}  we can get:
    \begin{equation}
        p(1-p)>q(1-q)
    \end{equation}
    Without loss of generality, we assume $\frac{1}{2}\leq p<q<1$. Therefore, according to the definition of likelihood, we know that 
    \begin{equation}l_{p2q}=P(mp,p)=\binom{m}{mp}q^{mp}(1-q)^{m-mp}
    \end{equation}
    \begin{equation}l_{q2p}=P(mq,q)=\binom{m}{mq}p^{mq}(1-p)^{m-mq}\end{equation}
    where $m\mathrm{,}mp\mathrm{,}mq$ are all positive integers. We need to prove the inequality:     
    \begin{equation}
    q^{mp}(1-q)^{m-mp}\frac{m!}{(m-mp)!(mp)!}<p^{mq}(1-p)^{m-mq}\frac{m!}{(m-mq)!(nq)!}
    \end{equation}
    Let the left side of the inequality be $L$ and the right side be $R$. 
    \begin{equation}
        L=q^{mp}(1-q)^{m-mp}\frac{m!}{(m-mp)!(mp)!},\quad R=p^{mq}(1-p)^{m-mq}\frac{m!}{(m-mq)!(mq)!}
    \end{equation}
    According to Stirling's formula~\cite{Stirling}, binomial coefficient of $L$ can be approximated as: 
    \begin{align}
        \frac{m!}{(m-mp)!(mp)!}&\approx\frac1{\sqrt{2\pi mp(1-p)}}\cdot\frac{m^m}{m^{mp}p^{mp}\cdot m^{m(1-p)}(1-p)^{m(1-p)}}\\&=\frac1{\sqrt{2\pi mp(1-p)}}\cdot\frac1{p^{mp}(1-p)^{m(1-p)}}
    \end{align}
    Substituting into $L$, it can be expressed as:
    \begin{align}
    L &\approx \frac{1}{\sqrt{2\pi mp(1-p)}} \cdot \frac{1}{p^{mp}(1-p)^{m(1-p)}} \cdot q^{mp}(1-q)^{m(1-p)} \\
  &= \frac{1}{\sqrt{2\pi mp(1-p)}} \cdot \left(\frac{q}{p}\right)^{mp} \left(\frac{1-q}{1-p}\right)^{m(1-p)}
\end{align}
    Rearranging the original expression:
    \begin{equation}
        L\approx \frac{1}{\sqrt{2\pi mp(1-p)}}\cdot\exp\left(mp\ln\frac{q}{p}+m(1-p)\ln\frac{1-q}{1-p}\right)
    \end{equation}
    Define \(\phi(s, t)=s\ln\frac{t}{s}+(1-s)\ln\frac{1-t}{1-s}\), then
        \begin{equation}
        L\approx \frac{1}{\sqrt{2\pi mp(1-p)}}\cdot\exp\left(m\phi(p,q)\right)
    \end{equation}
    In the same way, $R$ can be expressed as:
    \begin{equation}
        R\approx \frac{1}{\sqrt{2\pi mq(1-q)}}\cdot\exp\left(m\phi(q,p)\right)
    \end{equation}
    Consider the ratio $L/R$,
    \begin{equation}
        \frac{L}{R}\approx\frac{\frac{1}{\sqrt{2\pi mp(1-p)}}\exp\left(m\phi(p,q)\right)}{\frac{1}{\sqrt{2\pi mq(1-q)}}\exp\left(m\phi(q,p)\right)}=\sqrt{\frac{q(1-q)}{p(1-p)}}\cdot\exp\left(m\left(\phi(p,q)-\phi(q,p)\right)\right).
    \end{equation}
    When $\frac{1}{2}\leq p<q\leq1$, then 
    \begin{equation}
        \sqrt{\frac{q(1-q)}{p(1-p)}}<1
    \end{equation}
    For $exp\left(m\left(\phi(p,q)-\phi(q,p)\right)\right)$, we need to determine the sign of $\phi(p,q)-\phi(q,p)$. According to the definition of Kullback-Leibler divergence, we have
    \begin{equation}
        \begin{split}
            \phi(p,q)-\phi(q,p)&=p\ln\frac{q}{p}+(1-p)\ln\frac{1-q}{1-p}-q\ln\frac{p}{q}-(1-q)\ln\frac{1-p}{1-q}\\&=-(D_{KL}(p\|q)-D_{KL}(q\|p))
        \end{split}       
    \end{equation}
    Given the asymmetry of KL divergence~\cite{ElementsofInformationTheory}, when $\frac{1}{2}\leq p<q\leq1$, we have
    \begin{equation}
        D_{KL}(p\|q)-D_{KL}(q\|p)>0
    \end{equation}
    Therefore, $\phi(p,q)-\phi(q,p)<0$ holds. And $n$ is a positive integer, then 
    \begin{equation}
        exp\left(m\left(\phi(p,q)-\phi(q,p)\right)\right)<1
    \end{equation}
    In summary, as $m \to \infty$, the original expression equals:
    \begin{equation}
        \frac{L}{R}=\sqrt{\frac{q(1-q)}{p(1-p)}}\cdot\exp\left(m\left(\phi(p,q)-\phi(q,p)\right)\right) <1
    \end{equation}
    For small $n$ (where $mp$ amd $nq$ are integers), we can use Stirling's formula with a remainder term~\cite{RemarkStirling'sFormula}. 
    \begin{equation}
        m! = \sqrt{2\pi m}\left(\frac{m}{e}\right)^m e^{r_m}, \quad \text{where}, \quad \frac{1}{12m+1} < r_{m} < \frac{1}{12m}.
    \end{equation}
    It is easy to prove that, for the ratio \(L/R\), Stirling's formula with a remainder term has no effect. Because the remainder term $e^{r_m}$ is canceled out during the calculation of the ratio $L/R$. Clearly, when $m $ is a positive integer, \(L/R < 1\) ($l_{p2q}<l_{q2p}$)always holds. This contradicts our assumption that $l_{p2q}>l_{q2p}$. Therefore, the original proposition (lemma 1) holds.
    
\end{proof}

\subsubsection{Relationships between ${KLD_e(KLD_2)}$ and $FRD$}   
\label{Relationships}

The Kullback-Leibler divergence(KLD) measures the difference between two probability distributions $p$ and $q$. It is an asymmetric measure. For continuous distributions:
\begin{equation}
\label{KLD}
    D_\mathrm{KLD}(p, q)=\int p(x)\log\frac{p(x)}{q(x)}dx
\end{equation}
Among them, the base of $\log$ can be arbitrary (such as 2, $e$ or 10). The choice of the base does not affect the essential meaning of the KL divergence (i.e., the relative entropy or information gain between distributions), but it will change its numerical value and unit.

The base number does not affect the proof of the subsequent lemma.

\begin{lemma}
    Let $p(x|\theta)$ and $q=p(x|\theta + d\theta)$ be two close probability distributions (i.e., $\theta$ is continuous parameter $d\theta$ is an infinitesimal perturbation). For any distributions $p$ and $q$:
    \begin{equation}
        D_\mathrm{KLD}(p,q)\approx\frac{1}{2}D_{\text{FRD}}^2(p,q)
    \end{equation}
\end{lemma}

\begin{proof}
     Perform a Taylor expansion of $q = p(x|\theta + d\theta)$ at $\theta$: 
\begin{equation}
    \log p(x|\theta+d\theta)\approx\log p(x|\theta)+\sum_i\frac{\partial\log p}{\partial\theta_i}d\theta_i+\frac{1}{2}\sum_{i,j}\frac{\partial^2\log p}{\partial\theta_i\partial\theta_j}d\theta_id\theta_j
\end{equation}

Next, adjust the order: 
\begin{equation}
    \log p(x|\theta)-\log p(x|\theta+d\theta)\approx-\sum_i\frac{\partial\log p}{\partial\theta_i}d\theta_i-\frac{1}{2}\sum_{i,j}\frac{\partial^2\log p}{\partial\theta_i\partial\theta_j}d\theta_id\theta_j
\end{equation}

Substitute into the Eq.\ref{KLD}:
\begin{equation}
    D_\mathrm{KLD}(p,q)\approx\int p(x|\theta)\left(-\sum_i\frac{\partial\log p}{\partial\theta_i}d\theta_i-\frac{1}{2}\sum_{i,j}\frac{\partial^2\log p}{\partial\theta_i\partial\theta_j}d\theta_id\theta_j\right)dx
\end{equation}

The first item is calculated separately:
\begin{equation}
    -\int p(x|\theta)\sum_i\frac{\partial\log p}{\partial\theta_i}d\theta_i = -\frac{\partial}{\partial\theta_i}\int p(x|\theta)d\theta_i
\end{equation}

The integral of the probability density function $p(x|\theta)$ is equal to 1, so its derivative is 0. Then:
\begin{equation}
    D_\mathrm{KLD}(p,q)\approx-\frac{1}{2}\int p(x|\theta)\left(\sum_{i,j}\frac{\partial^2\log p}{\partial\theta_i\partial\theta_j}d\theta_id\theta_j\right)dx
\end{equation}

According to the relationship between Fisher information matrix $g_{ij}(\theta)$ and Hessian matrix $H_p$, then:
\begin{equation}
    g_{ij}(\theta)=-H_p\left[\frac{\partial^2\log p}{\partial\theta_i\partial\theta_j}\right]
\end{equation}

Therefore,
\begin{align}
D_\mathrm{KLD} (p,q)
&\approx -\frac{1}{2}\int p(x|\theta)\left(\sum_{i,j}\frac{\partial^2\log p}{\partial\theta_i\partial\theta_j}d\theta_id\theta_j\right)dx \\
&\approx -\frac{1}{2}\sum_{i,j}\frac{\partial^2\log p}{\partial\theta_i\partial\theta_j}d\theta_id\theta_j \cdot\int p(x|\theta)dx\\
&\approx\frac{1}{2}g_{ij}(\theta)d\theta_id\theta_j=\frac{1}{2}D_{\text{FRD}}^2(p,q)
\end{align}
\end{proof}

\subsubsection{The proof of $\chi^2$-divergence satisfying weak likelihood consistency
}
\label{wlc-CSD}
When considering three distributions, the magnitude relationship of estimated likelihoods may depend on the sample size. Therefore, when the magnitude relationship between $l_{p2q}$ and $l_{r2q}$ reverses under different values of $m$, strict LC (Likelihood Consistency) cannot be guaranteed(Similar to the situation in Fig. 1c)). However, in such cases, a weakened version of LC(WLC) need be defined.

According to the definition given in A7, we can propose the following lemma:



\begin{lemma}
    Given three 0-1 distributions $p=[p_1=q_1+e,q_2-e]$, $ q=[q_1,q_2]$ and $r=[q_1-e,q_2+e]$, where $q_1 + q_2 = 1$, $e>0$ is a small positive number (ensuring that all probability values are in the range $[0, 1]$, i.e., $0\leq e\leq \min(q_1, q_2)$), and $m$ is the sample size (compatible with the distributions $p, q ~and~r$). Let 
    \begin{itemize}
        \item $l_{p2q}$ is the likelihood when the true distribution is $q$ and the sampling distribution is $p$.
        \item $l_{r2q}$ is the likelihood when the true distribution is $q$ and the sampling distribution is $r$.
    \end{itemize}  
    The $\chi^2$ divergence ($d_{\chi^2}(\cdot,\cdot)$) can satisfy Axiom A7(WLC).
\end{lemma}

\begin{proof}
    According to the requirements in the definition of A7, it can be known that
    \begin{itemize}
        \item when $l_{p2q}\geq l_{r2q}$ (w.r.t. sample size $m$), then $d_{\chi^2}(p,q)\leq d_{\chi^2}(r,q)$.
        \item when $l_{p2q}\leq l_{r2q}$(w.r.t. sample size $m$), then $d_{\chi^2}(p,q)\geq d_{\chi^2}(r,q)$.
    \end{itemize}
     Clearly, it is necessary to ensure that under different sample sizes $m$, there exists \(d_{\chi^2}(p,q) =  d_{\chi^2}(r,q)\); then the $\chi^2$ distance can satisfy the requirements of WLC.
     
    When $p=[p_1=q_1+e,q_2-e]$, $ q=[q_1,q_2]$ and $r=[q_1-e,q_2+e]$, $d_{\chi^2}(p,q) $ and $ d_{\chi^2}(r,q)$ can be directly given by Eq.7 in the main text. 

    \begin{align*}
    {\rm{d}}_{\chi^2}(p,q) 
    &= \frac{(q_1 + e)^2}{q_1} + \frac{(q_2 - e)^2}{q_2} - 1 \\
    &= \frac{q_2(q_1^2 + 2q_1e + e^2) + q_1(q_2^2 - 2q_2e + e^2)}{q_1 q_2} - 1 \\
    &= \frac{(q_1 + q_2) (q_1 q_2 + e^2) }{q_1 q_2} - 1
    \end{align*}
    
    For any 0-1 distribution, $q_1+q_2=1$. Then:
    \begin{equation}
        {{\rm{d}}_{{\chi ^2}}}\left( {p,q} \right)=\frac{e^2}{q_1q_2}
    \end{equation}
    
    Similarly, we can obtain ${{\rm{d}}_{{\chi ^2}}}\left( {r,q} \right)$.  
    \begin{align*}
    {\rm{d}}_{\chi^2}(r,q) 
    &= \frac{(q_1 - e)^2}{q_1} + \frac{(q_2 + e)^2}{q_2} - 1 \\
    &= \frac{q_2(q_1^2 - 2q_1e + e^2) + q_1(q_2^2 + 2q_2e + e^2)}{q_1 q_2} - 1 \\
    &= \frac{q_1 q_2(q_1 + q_2) + e^2(q_1 + q_2)}{q_1 q_2} - 1 \\
    &= \frac{e^2}{q_1 q_2}
    \end{align*}
    Therefore,
    \begin{equation}
        d_{\chi^2}(p,q) = \frac{e^2}{q_1 q_2} = d_{\chi^2}(r,q)
    \end{equation}
    Since it has been proved in Lemma 1 that the $\chi^2$ divergence can satisfy the inequality relation in Axiom A2(LC). Combined with Eq.C2, the $\chi^2$ divergence can conforms to Axiom A7(WLC).
    \end{proof}

\end{document}